\documentclass{vldb-submission}

\usepackage[font=bf,skip=\baselineskip]{caption}
\usepackage{subcaption}
\usepackage{etoolbox}
\usepackage{graphicx,epsfig}
\usepackage{amsfonts,amsmath,amssymb,color,latexsym,multirow,xspace,color,url}
\usepackage[boxed]{algorithm}
\usepackage[noend]{algpseudocode}

\usepackage{balance}  

\newcommand{\lb}{\left\{}
\newcommand{\rb}{\right\}}
\newcommand{\st}{\,\middle|\,}

\newcommand{\lng}[1]{\ensuremath{\mathopen{|}#1\mathclose{|}}}

\renewcommand{\S}{\ensuremath{{\cal S}}}
\newcommand{\inp}{\ensuremath{{\cal I}}}
\newcommand{\dom}{\ensuremath{{\cal T}}}

\newcommand*\Let[2]{\State #1 $\gets$ #2}

\newcommand{\squishlist}{
   \begin{list}{$\bullet$}
    {
      \setlength{\itemsep}{0pt}
      \setlength{\parsep}{3pt}
      \setlength{\topsep}{3pt}
      \setlength{\partopsep}{0pt}
      \setlength{\leftmargin}{1.5em}
      \setlength{\labelwidth}{1em}
      \setlength{\labelsep}{0.5em} } }

\newcommand{\squishend}{
    \end{list}  }
\newcommand{\squishenum}{
   
   \begin{list}{scount}{ \usecounter{scount}}
    {
      \setlength{\itemsep}{0pt}
      \setlength{\parsep}{3pt}
      \setlength{\topsep}{3pt}
      \setlength{\partopsep}{0pt}
      \setlength{\leftmargin}{1.5em}
      \setlength{\labelwidth}{1em}
      \setlength{\labelsep}{0.5em} } }

\newcommand{\eat}[1]{}

\newtoggle{fullpaper}

\toggletrue{fullpaper}

\newtheorem{definition}{Definition}[section]

\newtheorem{example}{Example}[section]
\newtheorem{lemma}{Lemma}[section]

\newtheorem{theorem}[lemma]{Theorem}

\newtheorem{corollary}[lemma]{Corollary}
\newtheorem{claim}[lemma]{Claim}

        {\hspace*{\fill}$\Box$\par\vspace{4mm}}
        {\hspace*{\fill}$\Box$\par}

\newcommand{\stitle}[1]{\smallskip \noindent{\bf #1}}
\newcommand{\sstitle}[1]{\smallskip \noindent{\em #1}}

\newcommand{\nop}[1]{{}\xspace}
\newcommand{\hide}[1]{\hspace*{-5pt}\xspace}

\newcommand{\todo}[1]{{\color{red} #1}}

\newcommand{\ie}{{i.e.}\xspace}

\newcommand{\csec}{Section~}

\newcommand{\xxx}[1]{\xspace}

\newcommand{\norm}[1]{\ensuremath{\mathopen{\|}#1\mathclose{\|}}}

\def\blfootnote{\xdef\@thefnmark{}\footnotetext}

\begin{document}


\title{Design of Policy-Aware Differentially Private Algorithms}


\author{
Samuel Haney\\
	Duke University\\
  Durham, NC, USA\\
	\texttt{shaney@cs.duke.edu}
\and
Ashwin Machanavajjhala\\
  Duke University\\
  Durham, NC, USA\\
  \texttt{ashwin@cs.duke.edu}
\and
Bolin Ding\\
  Microsoft Research\\
  Redmond, WA, USA\\
  \texttt{bolin.ding@microsoft.com}
}
\maketitle

\begin{abstract}
The problem of designing error optimal differentially private algorithms is well studied. Recent work applying differential privacy to real world settings have used variants of differential privacy that appropriately modify the notion of neighboring databases. The problem of designing error optimal algorithms for such variants of differential privacy is open. In this paper, we show a novel transformational equivalence result that can turn the problem of query answering under differential privacy with a modified notion of neighbors to one of query answering under standard differential privacy, for a large class of neighbor definitions. 

We utilize the Blowfish privacy framework that generalizes differential privacy. Blowfish uses a {\em policy graph}  to instantiate different notions of neighboring databases. 
We show that the error incurred when answering a workload $\mathbf{W}$ on a database $\mathbf{x}$ under a  Blowfish policy graph  $G$ is identical to the error required to answer a transformed workload $f_G(\mathbf{W})$ on database $g_G(\mathbf{x})$ under standard differential privacy, where $f_G$ and $g_G$ are linear transformations based on $G$. Using this result, we develop error efficient algorithms for releasing histograms and multidimensional range queries under different Blowfish policies. We believe the tools we develop will be useful for finding mechanisms to answer many other classes of queries with low error under other policy graphs.
\end{abstract}

\section{Introduction}
\label{sec:intro}


The problem of private release of statistics from databases has become very important with the increasing use of databases with sensitive information about individuals in government and commercial organizations. $\epsilon$-Differential privacy \cite{icalp:Dwork06} has become the standard for private release of statistics due to its strong guarantee that ``similar'' inputs must yield ``similar'' outputs. Two input databases are similar if they are {\em neighbors}, meaning that they differ in the presence or absence of a single record. Output similarity is quantified by $\epsilon$, which bounds the log-odds of generating the same output from any pair of neighbors. Thus, if a record corresponds to all the data from one individual, differential privacy ensures that a single individual does not influence the inferences that can be drawn from the released statistics. Small $\epsilon$ results in greater privacy but also lesser utility. Thus, $\epsilon$ can be used to trade-off privacy for utility.

However, in certain applications (e.g., \cite{ashwin11:vldb}), the differential privacy guarantee  is too strict to produce private release of data that has any non-trivial utility. Tuning the parameter $\epsilon$ is not helpful here:  enlarging $\epsilon$ degrades the privacy guaranteed without a commensurate improvement in utility. Hence, recent work has considered relaxing differential privacy by modifying the notion of neighboring inputs by defining some metric over the space of all databases. Application designers can use this in addition to $\epsilon$ to better tradeoff privacy for utility. This idea has been applied to graphs (edge- versus vertex-differential privacy \cite{ashwin11:vldb}), streams (event- instead of individual-privacy \cite{stoc:DworkNPR10}), location privacy (geo-in\-dist\-inguishability \cite{andres2013geo}) and  to study fairness in targeted advertising (\cite{dwork2012fairness}), and has been formalized by multiple proposed frameworks (\cite{pets13:metric,blowfish,pods:KiferM12}). While these relaxations permit algorithms with significantly better utility than the standard notion of differential privacy, such algorithms must be designed from scratch. It is unknown how to derive algorithms that optimally leverage the relaxed privacy guarantee provided by the modified notion of neighbors to result in the least loss of utility. For instance, there are no known algorithms for releasing histograms or answering range queries with high utility under geo-indistinguishability. Moreover, there is no known method to utilize the literature on differentially private algorithms for this purpose. 

In this paper we present a novel and theoretically sound methodology for designing algorithms for relaxed privacy notions using algorithms that satisfy differential privacy, thus bridging the algorithm design problem under different privacy notions. Our results apply to the Blowfish privacy framework \cite{blowfish}, which generalizes differential privacy by allowing for different notions of neighboring databases (or {\em privacy policies}). We use this methodology to derive novel algorithms that satisfy Blowfish under relaxed privacy policies for releasing histograms and multi-dimensional range queries that have significantly better utility than the best known differentially private algorithms for these tasks. In the rest of this section, we present an overview of our results, describe the outline of the paper and then discuss related work. 

\stitle{Overview of Our Results.}
We first informally introduce the Blowfish privacy framework to help understand our theoretical and algorithmic results. 
The Blowfish framework instantiates a large class of ``similarity'' or neighbor definitions, using a ``policy graph'' defined over the domain of database records. Two input databases are neighbors if they differ in one record, and the differing values form an edge $(u,v)$ in the policy graph. Thus, one can not infer whether an individual's value was $u$ or $v$ based on the released output. For example, consider a {\em grid policy graph} (we will study its general form later in this paper): uniformly divide a 2D map into $k \times k$ grid cells, and each database record is one of the $k^2$ grid cells; only ``nearby'' points, e.g., pairs within Manhattan distance $\theta$, are connected by edges in the policy graph. Such a policy when used for location data implies that it is acceptable to reveal the rough location of an individual (e.g., the city), as two points belonging to two different cities are far away so no edge in the policy graph connects them; however, it requires that fine-grained location information (e.g., whether the individual is at home or at a nearby cafe) be hidden, when two grid points are close enough. This special instance of Blowfish framework is similar to a recently proposed notion called geo-indistinguishability \cite{andres2013geo}.  

Our main result is called \emph{transformational equivalence}, and we aim to show that a mechanism $\mathcal{M}$ for answering a set of linear queries $\mathbf{W}$ on a database $\mathbf{x}$ satisfies $(\epsilon, G)$-Blowfish privacy (i.e., differential privacy where neighboring databases are constructed with respect to the policy graph $G$) if and only if $\mathcal{M}$ is a mechanism for answering a transformed set of queries $f_G(\mathbf{W})$ on a transformed database $g_G(\mathbf{x})$ that  satisfies $\epsilon$-differential privacy. Here, $f_G$ and $g_G$ are linear transformations. However, we can not hope to prove such an equivalence in general. We prove (Theorem~\ref{thm:imposs}) that such an equivalence result implies that there exist a method to embed distances on any graph $G$ to distances in the $L_1$ metric without any distortion. This is because the distance between two input datasets induced by the neighborhood relation for differential privacy is the $L_1$ metric, and the distance under Blowfish is related to distances on graph $G$. Such a distortion free embedding is not known for large classes of graphs (e.g., a cycle) \cite{linial1995geometry}. 

Nevertheless, we are able to show this equivalence result for a large class of mechanisms and for a large class of graphs. First, we show that the transformational equivalence result holds for all algorithms that are instantiations of the matrix mechanism framework \cite{pods:LiHRMM10}. Matrix mechanisms algorithms, like Laplace mechanism for releasing histograms, and  hierarchical mechanism \cite{vldb:HayRMS10} and Privelet \cite{icde:XiaoWG10} for answering range queries, are popular building blocks for differentially private algorithm design. Transformational equivalence holds for such {\em data independent} mechanisms since the noise introduced by such mechanisms is independent of the input database. (The negative result uses a data-dependent mechanism whose error depends on the input database).

Next, we are also able to show that when $G$ is a tree there exist linear transformations $f_G$ and $g_G$ such that any mechanism $M$ for answering $\mathbf{W}$ on database $\mathbf{x}$  satisfies $(\epsilon, G)$-Blowfish privacy if and only if $M$ is an $\epsilon$-differentially private mechanism for answering $f_G(\mathbf{W})$ on database $g_G(\mathbf{x})$. The result follows (though not immediately) from the fact that trees permit a distortion free embedding into the $L_1$ metric \cite{fakcharoenphol2003tight, linial1995geometry}. This result holds for all privacy mechanisms, including data-dependent mechanisms. 

Finally, while the equivalence result does not hold for general mechanisms and general policy graphs, we can achieve an approximate equivalence. More specifically, we show the following \emph{subgraph approximation} result: if $G$ and $G'$ are such that every edge $(u,v)$ in $G$ is connected by a path of length at most $\ell$ in $G'$, then a mechanism $M$ that ensures $(\ell \cdot \epsilon, G')$-Blowfish privacy also ensures $(\epsilon, G)$-Blowfish privacy. Thus, if for a graph $G$ there is a tree $T$ such that distances in $G$ are not distorted by more than a multiplicative factor of $\ell$ in the tree $T$, then there exist transformations $f_T$ and $g_T$ such that an $\epsilon$-differentially private  mechanism $M$ for answering $f_T(\mathbf{W})$ on database $g_T(\mathbf{x})$, is also an $(\ell \cdot \epsilon, G)$-Blowfish private mechanism for answering $\mathbf{W}$ on  $\mathbf{x}$.

Additionally, a direct consequence of transformational equivalence is it allows us to derive error lower bounds and general approximation algorithms for Blowfish private mechanisms by extending work on error lower bounds for differentially private mechanisms \iftoggle{fullpaper}{\cite{stoc:HardtT10, bhaskara2012unconditional, nikolov2013geometry,edbt:chaoli13}}{\cite{stoc:HardtT10, edbt:chaoli13}}. We refer the reader to \iftoggle{fullpaper}{Appendix~\ref{sec:lowerbound}}{\todo{arxiv link}} for these results.

We apply the transformational equivalence theorems to derive novel  (near) optimal algorithms for answering multidimensional range query and histogram workloads under reasonable Blowfish policy graphs $G$ (like the grid graph). We reduce the problem of designing a Blowfish algorithm to that of  finding a differentially private mechanism for a new workload $\mathbf{W}_G = f_G(\mathbf{W})$ and a database $\mathbf{x}_g = g_G(\mathbf{x})$. We design matrix mechanism algorithms for all the policy graphs, and data dependent techniques when $G$ is a tree or can be approximated by a tree. For the policy graphs we consider, we show a polylogarithmic (in the domain size) improvement in error compared to the best data oblivious differentially private mechanism. We also present empirical results for the tasks of answering 1- and 2-dimensional range queries to show that our data dependent algorithms outperform their differentially private counterparts.

\stitle{Organization.}
The rest of this section is a brief survey of related work. Section~\ref{sec:background} presents notations and definitions that we will use throughout the paper. Section~\ref{sec:blowfish} introduces and motivates the Blowfish privacy framework. We describe our main result, transformational equivalence in Section~\ref{sec:transformationalEquivalence}. Section~\ref{sec:upperbound} presents novel mechanisms for answering multidimensional range queries and histogram queries  under various instantiations of the Blowfish framework, and presents the subgraph approximation lemma. In Section~\ref{sec:experiments}, we perform experiments comparing the performance of our Blowfish private mechanisms to differentially private mechanisms, and explore data-dependent mechanisms.  \iftoggle{fullpaper}{In the Appendix, Section~\ref{sec:lowerbound} gives examples of upper and lower bound results in differential privacy which extend to Blowfish privacy. Section~\ref{sec:multiple-components} extends our transformational equivalence results to policy graphs with multiple disconnected components.}
\ 

\stitle{Related Work.}
As mentioned earlier, works (\cite{andres2013geo},\cite{dwork2012fairness}) have developed relaxations of differential privacy that have specific applications (e.g. location privacy).
Other work has focused on developing flexible privacy definitions that generalize all these application specific notions. The Pufferfish framework \cite{pods:KiferM12} generalizes differential privacy by specifying what information should be kept secret, and the adversary's prior knowledge. He et al. \cite{blowfish} propose the Blowfish framework which also generalizes differential privacy and is inspired by Pufferfish. \cite{pets13:metric} investigates notions of privacy that can be defined as metrics over the set of databases. All these frameworks allow finer grained control on what information about individuals is kept secret, and what prior knowledge an adversary might possess, and thus allow customizing privacy definitions to the requirements of different applications.

As far as we aware, all previous work on relaxed privacy definitions have developed mechanisms directly for their applications.
We don't know of any work which shows how to map these relaxed privacy definitions to instances of differential privacy.

\section{Preliminaries}
\label{sec:background}

\begin{figure}[t]
\centering
{\small
\begin{equation*}
\left[\begin{array}{cccc}
1 & 0 & 0 & 0\\
0 & 1 & 0 & 0\\
0 & 0 & 1 & 0\\
0 & 0 & 0 & 1
\end{array}\right] \hspace{.5in}
\left[\begin{array}{cccc}
1 & 0 & 0 & 0\\
1 & 1 & 0 & 0\\
1 & 1 & 1 & 0\\
1 & 1 & 1 & 1
\end{array}\right]
\end{equation*}
}
\vspace{-5mm}
\caption{\label{fig:workloads} $\mathbf{I}_k$ (left) and $\mathbf{C}_k$ (right) workloads.}
\end{figure}

%
\stitle{Databases and Query Workloads.}
Let $\dom = \{v_1, v_2, \ldots,\allowbreak v_k\}$ be a domain of values with domain size $\lng{\dom} = k$. A database $D$ is a set of entries whose values come from $\dom$. Let $\inp_n$ be the set of all databases $D$ over $\dom$ such that the number of entries in $D$ is $n$, i.e., $\lng{D} = n$. And let $\inp$ be the set of all databases with any number of entries. We represent a database $D$ as a vector $\mathbf{x} \in \mathbb{R}^k$ with $\mathbf{x}[i]$ denoting the true count of entries in $D$ with the $i^{\rm th}$ value of the domain $\dom$.
That is, the database is represented as a histogram over the domain.
A {\em linear query} $\mathbf{q}$ is a $k$-dimensional row vector of real numbers with answer $\mathbf{q} \cdot \mathbf{x}$.
If the entries of $\mathbf{q}$ are restricted to 1's and 0's, we sometimes call $\mathbf{q}$ a linear counting query, since it counts the number of entries in $D$ with a particular subset of values in the domain.
A {\em workload} is a set of $q$ linear queries. So a workload can be represented as a $q\times k$ matrix $\mathbf{W}  = [\mathbf{q}_1 \mathbf{q}_2 \ldots \mathbf{q}_q]^\top \in \mathbb{R}^{q \times k}$, where each vector $\mathbf{q}_i \in \mathbb{R}^k$ corresponds to a linear query. The answer to the workload $\mathbf{W}$ is $\mathbf{W} \cdot  \mathbf{x}$, whose entries will be answers to the individual linear queries.

\begin{example}\label{ex:workloads}
Figure~\ref{fig:workloads} shows examples of two well studied workloads. $\mathbf{I}_k$ is the identity matrix representing the histogram query on $\dom$ reporting $[\mathbf{x}[1] \mathbf{x}[2] \ldots \mathbf{x}[k]]^\top$. $\mathbf{C}_k$ corresponds to the {\em cumulative histogram} workload, where each query corresponds to the {\em prefix sum}  $\sum_{j=1}^i \mathbf{x}[j]$.
\end{example}

%
\stitle{Differential privacy} is based on the concept of {\em neighbors}. Two databases are neighbors if they differ in one entry. 
\begin{definition}[Neighbors \cite{tcc:DworkMNS06}]\label{def:unbounded}
Any two databases $D$ and $D'$ are neighbors iff they differ in the presence of a single entry. That is, $\exists v \in \dom$, $D = D' \cup \lb v \rb$ or $D' = D \cup \lb v \rb$.
\end{definition}

An algorithm satisfies differential privacy if its outputs on any two neighboring databases are indistinguishable.

\begin{definition}[$\epsilon$-Differential privacy \cite{tcc:DworkMNS06}]\label{def:e-differential-privacy}
A randomized algorithm (mechanism) $\mathcal{M}$ satisfies $\epsilon$-differential privacy if for any subset of outputs $S \subseteq range(\mathcal{M})$,
and for any pair of neighboring databases $D$ and $D'$, 
\begin{equation*}
\mathrm{Pr} [\mathcal{M}(D) \in S] \le e^\epsilon \cdot \mathrm{Pr} [\mathcal{M}(D') \in S].
\end{equation*}
\end{definition}

A slightly different definition of neighbors yields a common variant of differential privacy: one database can be obtained from its neighbor by replacing one entry $x$ with a different value $y \in \dom$. The resulting privacy notation is called {\em $\epsilon$-indistinguishability} or {\em bounded $\epsilon$-differential privacy}. Unless otherwise specified, we use the term differential privacy to mean the original {\em unbounded} version (Definitions~\ref{def:unbounded}-\ref{def:e-differential-privacy}). 

\stitle{Sensitivity and Private Mechanisms.}
Suppose we use a differentially private algorithm $\mathcal{M}$ to publish the result of a workload $\mathbf{W}$ on a database $D$ that is represented as a vector $\mathbf{x}$. $\mathcal{M}$ is called {\em data independent} if the amount of noise $\mathcal{M}$ adds does not depend on the database $\mathbf{x}$, and {\em data dependent} otherwise. In both cases, the amount of noise depends the {\em sensitivity} of a workload.
$||\cdot||_1$ denotes the $L_1$ norm.

\begin{definition}[Sensitivity \cite{tcc:DworkMNS06,edbt:chaoli13}]
  \label{def:sensitivity-diff-priv}
Let $\cal N$ denote the set of pairs of neighbors. The $L_1$ sensitivity of $\mathbf{W}$ is: 
\begin{equation*}
\Delta_{\mathbf{W}} \ = \ \max_{(\mathbf{x},\mathbf{x}')\in \cal N} \norm{\mathbf{Wx} - \mathbf{Wx'}}_1.
\end{equation*}
\end{definition}

\begin{example}
The  $L_1$ sensitivities of $\mathbf{I}_k$ and $\mathbf{C}_k$ are $1$ and $k$, resp.  
\end{example}

A well-studied class of differentially private algorithms is called {\em Laplace mechanism} \cite{tcc:DworkMNS06}. Let $\mathrm{Lap}(\sigma)^m$ be a $m$-dimensional vector of independent samples, where each sample is drawn from $\eta \propto \exp(-\frac{|x|}{\sigma})$. 

%
\stitle{Measuring Errors.}
We use {\em mean squared error} to measure the amount of noise injected in private algorithms.

\begin{definition}[Error]
  Let $\mathbf{W} = [\mathbf{q}_1\mathbf{q}_2\ldots\mathbf{q}_q]^\top$ be a workload of linear queries, and $\mathcal{M}$ be a mechanism to publish the query result privately. Let $\mathbf{x}$ be the vector representing the database. The mean squared error of answering a workload $\mathbf{W}$ on the database $\mathbf{x}$ using $\mathcal{M}$ is
  \begin{equation*}
    \mathrm{ERROR}_\mathcal{M}(\mathbf{W}, \mathbf{x}) 
 = \sum_{i=1}^q \mathbb{E}\left[ (\mathbf{q_ix} - \mathcal{M}(\mathbf{q_i},\mathbf{x}))^2 \right]
  \end{equation*}
  where $\mathcal{M}(\mathbf{q}, \mathbf{x})$ is the noisy answer of query $\mathbf{q}$.  We define the data-independent error of a mechanism $\mathcal{M}$ to be
  \begin{equation*}
    \mathrm{ERROR}_\mathcal{M}(\mathbf{W}) = \max_\mathbf{x} \left\{ \mathrm{ERROR}_\mathcal{M}(\mathbf{W},\mathbf{x}) \right\}.
  \end{equation*}
  \label{def:error}
\end{definition}

Laplace mechanism is known to provide $\epsilon$-differential privacy with mean squared error as a function of $L_1$ sensitivity.

\begin{theorem}[\cite{tcc:DworkMNS06}]\label{thm:lap-gaussian-error}
Let $\mathbf{W}$ be a $q\times k$ workload. The Laplace mechanism $\mathcal{L}(\mathbf{W},\mathbf{x}) = \mathbf{W}\mathbf{x} + \mathrm{Lap}(\sigma)^q$ satisfies $\epsilon$-differential privacy, with $\mathrm{ERROR}_\mathcal{L}(\mathbf{W}) = 2q \Delta_{\mathbf{W}}^2/\epsilon^2$.
\end{theorem}

\eat{
\begin{definition}[Error]
  Let $\mathbf{W} = [\mathbf{q}_1\mathbf{q}_2\ldots\mathbf{q}_q]^\top$ be a workload of linear queries, and $\mathcal{M}$ be a mechanism to publish the query result privately. Let $\mathbf{x}$ be the vector representing the database. The mean squared error of answering a linear query $\mathbf{q}$ on the database $\mathbf{x}$ using $\mathcal{M}$ is
  \begin{equation*}
    \mathrm{ERROR}_\mathcal{M}(\mathbf{q}, \mathbf{x}) = \mathbb{E}\left[ (\mathbf{qx} - \mathcal{M}(\mathbf{q},\mathbf{x}))^2 \right],
  \end{equation*}
  where $\mathcal{M}(\mathbf{q}, \mathbf{x})$ is the noisy answer of query $\mathbf{q}$. The mean squared error of answering the workload $\mathbf{W}$ on $\mathbf{x}$ is
  \begin{equation*}
    \mathrm{ERROR}_\mathcal{M}(\mathbf{W}, \mathbf{x}) = \sum_{i=1}^q \mathrm{ERROR}_\mathcal{M}(\mathbf{q}_i, \mathbf{x}).
  \end{equation*}
  Finally, we define the data-independent error to be
  \begin{equation*}
    \mathrm{ERROR}_\mathcal{M}(\mathbf{W}) = \max_\mathbf{x} \left\{ \mathrm{ERROR}_\mathcal{M}(\mathbf{W},\mathbf{x}) \right\}.
  \end{equation*}
  \label{def:error}
\end{definition}
}

\section{Blowfish Privacy}
\label{sec:blowfish}

The {\em Blowfish privacy} framework, originally introduced by He et al. \cite{blowfish}, is a class of privacy notations that generalize neighboring databases in differential privacy. It allows privacy policy to focus only on neighbors that users are sensitive about. The major building block of an instantiation of Blowfish is called {\em policy graph}. A policy graph encodes users' private and sensitive information by specifying which pairs of domain values in $\dom$ should {\em not} be distinguished between by an adversary. By carefully choosing a policy graph (or equivalently, restricting the set of neighboring databases), Blowfish trades-off privacy for potential gains in utility.

\begin{definition}[Policy graph]
A policy graph is a graph $G = (V,E)$ with $V \subseteq \dom \cup \left\{ \bot \right\}$, where $\bot$ is the name of a special vertex, and $E \subseteq (\dom \cup \left\{ \bot \right\}) \times (\dom \cup \left\{ \bot \right\})$.
\end{definition}

The above definition of policy graph is slightly different from the one in \cite{blowfish} with an additional special vertex $\bot$ to generalize both unbounded and bounded versions of differential privacy.
Intuitively, an edge $(u,v) \in E$ defines a pair of domain values that an adversary should not be able to distinguish between. $\bot$ is a dummy value not in $\dom$, and an edge $(u, \bot) \in E$ means that an adversary should not be able to distinguish between the presence of a tuple with value $u$ or the absence of the tuple from the database.
For technical reasons, if there is some edge incident on $\bot$, we add a zero column vector $\mathbf{0}$ into the workload $\mathbf{W}$ to correspond to the dummy value $\bot$, as well as a zero entry in the database vector $\mathbf{x}$ correspondingly. So it is ensured that every node in $V$ is associated with a column in $\mathbf{W}$ and an entry in $\mathbf{x}$.

We next revisit the Blowfish privacy framework.
%

\begin{definition}[Blowfish neighbors]\label{def:neighbors-blowfish}
  We consider a policy graph $G=(V,E)$. Let $D$ and $D'$ be two databases. $D$ and $D'$ are neighbors, denoted $(D, D') \in {\cal N}(G)$, iff exactly one of the following is true:
\squishlist
  \item $D$ and $D'$ differ in the value of exactly one entry such that $(u, v)\in E$, where $u$ is the value of the entry in $D$ and $v$ is the value of the entry in $D'$;
  \item $D$ differs from $D'$ in the presence or absence of exactly one entry, with value $u$, such that $(u, \bot) \in E$.
\squishend
\end{definition}

\begin{definition}[$(\epsilon,G)$-Blowfish Privacy]\label{def:e-blowfish}
Let $G$ be a policy graph.  A mechanism $\mathcal{M}$ satisfies $(\epsilon,G)$-Blowfish privacy if for any subset of outputs $S \subseteq range(\mathcal{M})$, and for any pair of neighboring databases $(D, D') \in {\cal N}(G)$, 
\begin{equation*}
\mathrm{Pr} [\mathcal{M}(D) \in S] \le e^\epsilon \cdot \mathrm{Pr} [\mathcal{M}(D') \in S].
\end{equation*}
\end{definition}

\stitle{Policy graph and Privacy guarantee.}
Policy graphs are used to define how privacy will be guaranteed to users, independent of an adversary's knowledge.
For example, when the users require the strongest privacy guarantee, the following two policy graphs can be used,
%
\begin{equation*}
  G = (V,E) \text{ such that } E = \left\{ (u,\bot) \st \forall u \in \dom \right\}, ~\hbox{and}
\end{equation*}
\begin{equation*}
  G = (V,E) \text{ such that } E = \left\{ (u,v) \st \forall u,v \in \dom\right\},
\end{equation*}
which corresponds to the unbounded and bounded versions of differential privacy, respectively.
More generally, if a policy graph does not include $\bot$, we are essentially focusing on databases from $\inp_n$, \ie, databases with fixed known size. 



When the privacy guarantee is relaxed, we can adjust the policy graphs so that our algorithms in \csec\ref{sec:upperbound} will have higher utility. Following are two examples of designing such policy graphs for real-life scenarios.
%

({\em Line Graph}) Consider a totally ordered domain $\dom = \{a_1, a_2, \ldots, a_k\}$, where $\forall i: ~ a_i < a_{i+1}$. One such example is a database of binned salaries of individuals, where $a_i$ corresponds to a salary between $2^{i-1}$ and $2^i$. When only revealing rough ranges of salaries is fine, it is OK for an adversary to distinguish between values that are far apart (e.g., $a_1$ vs $a_k$) but not distinguish between values that are closer to each other. So we can use a line graph as our policy graph to express this guarantee, where only adjacent domain values $a_i$ and $a_{i+1}$ are connected by an edge.

({\em Grid Graph}) In the scenario of location data, when revealing rough location information is fine but more precise location information is private, we can use a grid policy graph (also discussed in \csec\ref{sec:intro}). Here, a 2D map uniformly divided into $k \times k$ grid points as nodes in the graph (i.e., $\dom = \{1, \ldots, k\} \times \{1, \ldots, k\}$).  Only ``nearby'' points are connected by edges: $E = \{ (u,v) |  d(u,v) \leq \theta \}$ where $d(u,v)$ denotes the distance between two points $u, v \in \dom$ (e.g., Manhattan distance) and $\theta$ is a policy-specific parameter. 
%
%

\stitle{Metric on databases.}
In general, a policy graph introduces a metric over databases, which quantifies the privacy guarantee provided by Blowfish. Consider two databases that differ in one tuple: $D_1 = D \cup \left\{ u \right\}$ and $D_2 = D \cup \left\{ v \right\}$, define the distance between $D_1$ and $D_2$ be $\mathrm{dist}_G(u,v)$, i.e., the length of the shortest path between $u$ and $v$ in $G$.
For mechanism ${\cal M}$ satisfying $(\epsilon, G)$-Blowfish privacy (Defs~\ref{def:neighbors-blowfish}-\ref{def:e-blowfish}), we have,
\begin{equation}\label{equ:metric}
\mathrm{Pr} [\mathcal{M}(D_1) \in S] \le  e^{\epsilon \cdot \mathrm{dist}_G(u,v)} \cdot \mathrm{Pr} [\mathcal{M}(D_2) \in S].
\end{equation}
For two databases differing in more than one tuple, we can repeatedly apply \eqref{equ:metric} for each differing tuple.

For the grid graph policy, changing the location of a tuple from $u$ to $v$ results in the output probabilities of $\cal M$ differing by a factor of $e^\epsilon$ if $d(u,v) \leq \theta$, and differing be a factor of $e^{\epsilon \cdot \lceil d(u,v) / \theta \rceil}$ in general. So finer grained location information gets stronger protection. This privacy guarantee is identical to a recent notion called geo-indistinguishability \cite{andres2013geo}.
%

In this paper, we assume Blowfish policy graphs are connected. We discuss Blowfish policies with disconnected components in \iftoggle{fullpaper}{Appendix~\ref{sec:multiple-components}}{the full version}.

\section{Transformational Equivalence}
\label{sec:transformationalEquivalence}
\newcommand{\toappendix}[1]{In Appendix}

We now present our main result, called {\em transformational equivalence}, which we will use to design Blowfish private algorithms later in the paper. This result establishes a mechanism-preserving two way relationship between Blowfish privacy and differential privacy.
In general, our transformation can be stated as follows:
\textit{For policy graph $G$, there exists a transformation of the workload and database, $(\mathbf{W},\mathbf{x}) \rightarrow (\mathbf{W}_G,\mathbf{x}_G)$ such that $\mathbf{Wx} = \mathbf{W}_G\mathbf{x}_G$, and a mechanism $\mathcal{M}$ is an $(\epsilon, G)$-Blowfish private mechanism for answering workload $\mathbf{W}$ on input $\mathbf{x}$ if and only if $\mathcal{M}$ is also an $\epsilon$-differentially private mechanism for answering $\mathbf{W}_G$ on $\mathbf{x}_G$.}
However, we can't hope to show this result in general. We prove that for certain mechanisms transformational equivalence holds only when distances on a graph can be {\em embedded} into points in $L_1$ with no distortion. It is well known \cite{linial1995geometry} that not all graphs permit such embeddings. 

Hence, we show different results for different restrictions on $\mathcal{M}$ and $G$.
In Section~\ref{sec:transequiv-mm}, we show that under a class of mechanisms called the \emph{matrix mechanism}, transformational equivalence holds for any policy graph.
In Section~\ref{sec:transequiv-tree}, we show via a metric embedding-like argument that when $G$ is a tree, transformational equivalence holds for any mechanism $\mathcal{M}$.
In Section~\ref{sec:transequiv-general}, we state the negative result for general graphs  and mechanisms, and present an approximate transformational equivalence that uses spanning trees of $G$ albeit with some loss in the utility.
All of our results rely on the existence of a transformation matrix $\mathbf{P}_G$ with certain properties, whose construction is discussed in  Section~\ref{sec:construction-PG}

\subsection{Equivalence for  Matrix Mechanism}\label{sec:transequiv-mm}

Li et al \cite{pods:LiHRMM10} describe the {\em matrix mechanism} framework for optimally answering a workload of linear queries. The key insight is that while some workloads $\mathbf{W}$ have a high sensitivity, they can be answered with low error by answering a different {\em strategy} query workload $\mathbf{A}$ such that (a) $\mathbf{A}$ has a low sensitivity $\Delta_\mathbf{A}$, and (b) rows in $\mathbf{W}$ can be reconstructed using a small number of rows in $\mathbf{A}$. 

In particular, let $\mathbf{A}$ be a $p \times k$ matrix, and $\mathbf{A}^{+}$ denote its Moore-Penrose pseudoinverse, such that $\mathbf{W}\mathbf{A}\mathbf{A}^{+} = \mathbf{W}$. The matrix mechanism is given by the following: 
\begin{equation}\label{eqn:mech}
    \mathcal{M}_\mathbf{A}(\mathbf{W},\mathbf{x}) = \mathbf{Wx} + \mathbf{WA}^+Lap(\Delta_{\mathbf{A}}/\epsilon)^p
\end{equation}
where, $Lap(\lambda)^p$ denotes $p$ independent random variables drawn from the Laplace distribution with scale $\lambda$. 
Recall that $\Delta_{\mathbf{A}}$ is the sensitivity of workload $\mathbf{A}$. 
It is easy to see that all matrix mechanism algorithms are data independent (i.e., the noise is independent of the input dataset).

In order to extend matrix mechanisms to Blowfish, we define the Blowfish specific sensitivity of a workload, $\Delta_{\mathbf{W}}(G)$ analogously to Definition~\ref{def:sensitivity-diff-priv}:
\begin{definition}
    \label{def:sensitivity-blowfish}
  The $L_1$ policy specific sensitivity of a query matrix $\mathbf{W}$ with respect to policy graph $G$ is
  \begin{equation*}
    \Delta_{,\mathbf{W}}(G) = \max_{(\mathbf{x},\mathbf{x}')\in N(G)}\norm{\mathbf{Wx} - \mathbf{Wx}'}_1
  \end{equation*}
\end{definition}

%

Let $P_G$ be a matrix that satisfies the following properties. We will describe its construction in Section~\ref{sec:construction-PG}.
\squishlist
  \item $\mathbf{P}_G$ has $|V| - 1$ rows and $|E|$ columns.
  \item Let $\mathbf{W}_G = \mathbf{WP}_G$. Then $\Delta_\mathbf{W}(G) = \Delta_{\mathbf{W}_G}$. I.e., the sensitivity of workload $\mathbf{W}$ under Blowfish policy $G$ is the same as the sensitivity of $\mathbf{W}_G$ under differential privacy.
  \item $\mathbf{P}_G$ has full row rank (and therefore a right inverse $\mathbf{P}_G^{-1}$). For vector $\mathbf{x}$ we let $\mathbf{x}_G$ denote $\mathbf{P}_G^{-1}\mathbf{x}$.
\squishend

Given such a $\mathbf{P}_G$, we can show our first transformational equivalence result.

\begin{theorem}
  \label{thm:mm-trans-equiv}
  Let $G$ be a Blowfish policy graph and $\mathbf{W}$ be a workload. Suppose $\mathbf{P}_G$ exists with the properties given above. Then the matrix mechanism given by Equation~\ref{eqn:mech} is both a $(\epsilon,G)$-Blowfish private mechanism for answering $\mathbf{W}$ on $\mathbf{x}$ and an $\epsilon$-differentially private algorithm for answering $\mathbf{W}_G$ on $\mathbf{x}_G$. Since $\mathbf{Wx}=\mathbf{W}_G \mathbf{x}_G$, the mechanism has the same error in both instances.
\end{theorem}
{\sc Proof.}
  We show that
  \begin{equation*}
    \mathbf{Wx} + \mathbf{WA}^+Lap(\frac{\Delta_{\mathbf{A}}(G)}{\epsilon})^p = \mathbf{W}_G\mathbf{x}_G + \mathbf{W}_G\mathbf{A}_G^+Lap(\frac{\Delta_{\mathbf{A}_G}}{\epsilon})^p.
  \end{equation*}
  First,
  \begin{equation*}
    \mathbf{W}\mathbf{P}_G\mathbf{P}_G^{-1}\mathbf{x} = \mathbf{W}\mathbf{I}_k\mathbf{x} = \mathbf{W}\mathbf{x}.
  \end{equation*}
  Next, by assumption we have that $\Delta_{\mathbf{A}}(G) = \Delta_{\mathbf{A}_G}$.
  Finally,
  \begin{align*}
    \mathbf{W}_G \mathbf{A}_G^+ &= \mathbf{WP}_G(\mathbf{AP}_G)^+ \\
    &= \mathbf{WP}_G\mathbf{P}_G^+\mathbf{A}^+ \\
    &= \mathbf{WA}^+ \mbox{\hspace{10mm}($\mathbf{P}_G$ has full row rank) \rlap{$\qquad \Box$}} 
  \end{align*}

\subsection{Equivalence when G is a Tree}\label{sec:transequiv-tree}

When $G$ is a tree, we can show something stronger, that transformational equivalence holds for any mechanism $\mathcal{M}$.
More formally, suppose $\mathbf{P}_G$ has the following property.
\begin{claim}\label{claim:tree-neigh}
  If $G$ is a tree, any pair of $\mathbf{y}, \mathbf{z} \in \mathbb{R}^k$ are neighbors according to the Blowfish policy $G$ if and only if $\mathbf{P}_G^{-1}\mathbf{y}$ and $\mathbf{P}_G^{-1}\mathbf{z}$ are neighbors according to unbounded differential privacy (which are vectors with $L_1$ distance of 1).
\end{claim}

We construct a $\mathbf{P}_G$ satisfying Claim~\ref{claim:tree-neigh} in Section~\ref{sec:construction-PG}.
Our stronger transformational equivalence result follows.

\begin{theorem}\label{thm:blowfishEquality}
  Let $\mathbf{x} \in \mathbb{R}^k$ represent a database, $\mathbf{W}$ be a workload with $q$ linear queries, and $G = (V,E)$ be a Blowfish policy graph and a tree. We can find an invertible mapping  given by $f(\mathbf{x},\mathbf{W},G) = (\mathbf{P}_G^{-1}\mathbf{x},\mathbf{W}\mathbf{P}_G$), where $\mathbf{P}_G$ is a matrix depending on $G$, such that $\mathcal{M}$ is a $(G,\epsilon)$-Blowfish private mechanism for answering $(\mathbf{W},\mathbf{x})$ with error $\alpha$ if and only if $\mathcal{M}$ is an $\epsilon$-differentially private mechanism for answering $(\mathbf{W}\mathbf{P}_G,\mathbf{P}_G^{-1}\mathbf{x})$ with error $\alpha$.
\end{theorem}

\begin{proof}
  Suppose $\mathbf{P}_G$ satisfies the properties given at the beginning of the section.
  Then, mechanism $\mathcal{M}$ will have the same error on both instances, since the true answers to the workloads are the same in both cases:
\begin{equation*}
  \mathbf{W}\mathbf{P}_G\mathbf{P}_G^{-1}\mathbf{x} = \mathbf{W}\mathbf{I}_k\mathbf{x} = \mathbf{W}\mathbf{x}.
\end{equation*}
Additionally, the mapping is invertible, since $\mathbf{W}\mathbf{P}_G\mathbf{P}_G^{-1} = \mathbf{W}$ and $\mathbf{P}_G\mathbf{P}_G^{-1}\mathbf{x} = \mathbf{x}$.
$\mathcal{M}$ is both $(\epsilon,G)$-differentially private on $\mathbf{W,x}$ and $\epsilon$-differentially private on $\mathbf{W}_G, \mathbf{x}_G$, since the mapping preserves neighbors.
\end{proof}

\subsection{Equivalence for General Graphs and Mechanisms}\label{sec:transequiv-general}
We can show that we can not hope to prove transformation equivalence for general graphs and mechanisms. First we define an {\em embedding} of graphs. 
\begin{definition}
Let $G = (V,E)$ be a graph. Let $\rho$ be a deterministic mapping from vertices in $V$ to real valued vectors. Let $d_G(u,v)$ denote the shortest distance between vertices $u$ and $v$, and $d(\rho(u), \rho(v)) = ||\rho(u) - \rho(v)||_1$ the $L_1$ distance between the mapped vectors. We define the {\em stretch} of mapping $\rho$ to be $\max_{u,v \in V} d(\rho(u), \rho(v))/d_G(u,v)$, or the maximum multiplicative increase in distances due to the mapping. Similarly the {\em shrink} of $\rho$ is defined as $\min_{u,v \in V}$ $d(\rho(u), \rho(v))/d_G(u,v)$, or the smallest multiplicative decrease in distances. We call $\rho$ an isometric embedded if stretch and shrink equal to $1$. 
\end{definition}

We now show that for graphs with no isometric embedding into points in $L_1$, transformational equivalence does not hold. It is well known that such graphs exist. One example is the cycle on $n$ vertices, for which no deterministic mapping is known with stretch less than $(n-1)$ \cite{linial1995geometry}.  
\begin{theorem}\label{thm:imposs}
Let $G$ be a graph that does not have an isometric embedding into points in $L_1$. There exists a mechanism $\mathcal{M}$ and workload $\mathbf{W}$ such that for any transformation of $(\mathbf{W}, \mathbf{x}) \rightarrow (\mathbf{W}_G, \mathbf{x}_G)$ such that $\mathbf{W}\mathbf{x} = \mathbf{W}_G\mathbf{x}_G$, either $\mathcal{M}$ is not an $(\epsilon, G)$-Blowfish private mechanism for answering $\mathbf{W}$ on $\mathbf{x}$, or $\mathcal{M}$ is not a $\epsilon$-differentially private mechanism for answering $\mathbf{W}_G$ on $\mathbf{x}_G$.
\end{theorem}
We refer the reader to \iftoggle{fullpaper}{Appendix~\ref{sec:proofs-general-transequiv}}{the full paper} for all proofs in this section. We would like to note that transformational equivalence holds for policy graphs that are trees, since trees can be isometrically embedded into points in $L_1$, and the $\mathbf{P}_G$ we construct is one such mapping. Moreover, the proof for Theorem~\ref{thm:imposs} requires a mechanism $\mathcal{M}$ that is data-dependent; it uses the exponential mechanism that introduces noise that depends on the input. We believe data dependence is necessary for the negative result, and hence we were able to show transformational equivalence for matrix mechanism algorithms (that are data independent). 

Despite the negative result, we next show an approximate transformational equivalence for general graphs and mechanisms with some loss in utility. 
The error in our approximate transformation is proportion to the stretch resulting from embedding $G$ into a spanning tree of $G$, $G'$.
Transformational equivalence can then be applied on $G'$ giving us an approximate equivalence under the original graph.

\begin{lemma}(Subgraph Approximation)\label{lem:subgraph} 
  Let $G = (V,E)$ be a policy graph. Let $G' = (V, E')$ be a spanning tree of $G$ on the same set of vertices, such that every $(u,v) \in E$ is connected in $G'$ by a path of length at most $\ell$ ($G'$ is said to be an $\ell$-approximate subgraph\footnote{While we that require $V(G) = V(G')$, the proof does not require $G'$ to be a subgraph of $G$ (i.e., $E' \subseteq E$). But it suffices for the applications of this technique in this paper.}). Then for any mechanism $\mathcal{M}$ which satisfies $(\epsilon,G')$-Blowfish privacy, $\mathcal{M}$ also satisfies $(\ell \cdot \epsilon,G)$-Blowfish privacy.
\end{lemma}
\eat{
\begin{proof}
  Assume $D$ and $D'$ are neighboring databases under policy graph $G$. Then $D = A \cup \left\{ x \right\}$ and $D' = A \cup \left\{ y \right\}$ for some database $A$, and $(x,y) \in E$. From our assumption, $x$ and $y$ are connected by a path in $G'$ of length at most $\ell$. Therefore, there exist a sequence of vertices $x=v_1,\dots,v_j = y$ such that $(v_i, v_{i+1}) \in E$ and $j<\ell$. Further, $A \cup \left\{ v_i \right\}$ and $A \cup \left\{ v_{i+1} \right\}$ are neighbors under policy graph $G'$. Therefore, 
  \begin{equation*}
    \mathrm{Pr} [\mathcal{M}(A \cup \left\{ v_i \right\}) \in S] \le e^\epsilon \cdot \mathrm{Pr} [\mathcal{M}(A \cup \left\{ v_{i+1} \right\}) \in S].
  \end{equation*}
    Composing over all $1 \le i \le j$ gives us the desired result. 
\end{proof}
}

\begin{corollary}
  \label{cor:approx-equivalence}
  Let $G$ be a graph and let $G'$ be an $\ell$-approximate spanning tree.
  Suppose $\mathcal{M}$ is an $\epsilon$ differentially private mechanism for $\mathbf{W}_{G'},\mathbf{x}_{G'}$.
  Then, $\mathcal{M}$ is an $(\ell \cdot \epsilon,G)$-Blowfish private mechanism for $\mathbf{W},\mathbf{x}$.
  Since $\mathbf{Wx}=\mathbf{W}_G \mathbf{x}_G$, the mechanism has the same error in both instances.
\end{corollary}

A well-known result of Fakcharoenphol et al \cite{fakcharoenphol2003tight} (Theorem~2) shows that any metric can be embedded into a distribution of trees with $O(\log n)$ expected {stretch}.
It would be desirable to use this result to give a $O(\log(n))$-approximate subgraph for any graph $G$.
However, because the bound on {stretch} only holds in expectation, our privacy guarantee would only hold in expectation!
A deterministic embedding with low stretch does not always exist.
To see this, consider an $n$-vertex cycle.
Any spanning tree consists of all but one edge $(u,v)$ from the cycle. While $u$ and $v$ were distance 1 apart in the cycle, they are distance $n-1$ apart in the spanning tree! If we picked a spanning tree at random by randomly choosing the edge that was dropped the expected stretch is only $2$.
Using a union bound over all pairs in this example, we can also see that it is impossible to guarantee a low stretch (and therefore a privacy guarantee) with high probability.
Therefore, we cannot apply Lemma~\ref{lem:subgraph} in a general way to find a suitable spanner for any policy graph $G$.
However, the lemma is still useful in many cases and we will use it throughout the rest of the paper.

\subsection{Construction of $\mathbf{P}_G$}\label{sec:construction-PG}

Our construction of $\mathbf{P}_G$ from the policy graph $G$ is related to the vertex-edge incidence matrix, where every row   corresponds to a vertex in $G$, every column corresponds to an edge in $G$. A column has two non-zero entries (1 and -1) in the rows corresponding vertices connected by the corresponding edge. We can view $\mathbf{P}_G$ and $\mathbf{P}_G^{-1}$ as linear transformations from the vertex domain $V$ to the edge domain $E$. While $\mathbf{x}$ corresponds to counts on vertices of $G$, the transformed database $\mathbf{x}_G = \mathbf{P}_G^{-1}\mathbf{x}$ would assign weights to edges in $G$. Similarly, while an original linear query $q \in \mathbf{W}$ associates weights on (a subset of) vertices in $G$, a query $q_G \in \mathbf{W}_G = \mathbf{W} \mathbf{P}_G$ associates weights on (a subset of) edges in $G$. This intuition will be very useful when using the equivalence result to design $(\epsilon, G)$-Blowfish algorithms. 

We cannot just use the vertex-edge incidence matrix as $\mathbf{P}_G$ since it may either have $k+1$ rows (when $G$ contains $\bot$), or since it does not have an inverse (when $G$ does not contain $\bot$). We will describe our construction of $\mathbf{P}_G$ that satisfies all our constraints in the rest of the section. The details are quite technical, and an uninterested reader can skip over them and still understand the rest of the paper. As mentioned before, we assume $G$ is connected. Our constructions also extend to policy graphs that are disconnected  and are discussed in \iftoggle{fullpaper}{Appendix~\ref{sec:multiple-components}}{the full paper}.

\subsubsection*{Case I: Unbounded (with $\bot$)}
 
We start our construction with a simple case. Let $G=(V,E)$ be a {\em connected} undirected graph, with $V = \dom \cup \{\bot\}$,  $|\dom| = k$. We define $\mathbf{P}_G$ to be the following  $k \times |E|$-matrix:  let each row of $\mathbf{P}_G$ correspond to a value in $\dom$; for each edge $(u,v) \in E$ ($u, v \neq \bot$), add a column to $\mathbf{P}_G$ with a $1$ in the row corresponding to value $u$, a $-1$ in the row corresponding to value $v$ (order of $1$ and $-1$ is not important), and zeros in the rest of the rows; and for each edge $(u,\bot) \in E$ ($u \neq \bot$), add a column with a $1$ in the row corresponding to $u$ and zeros in the rest. Figure~\ref{fig:graphs} gives an example.

\begin{figure}[t]
\centering
 \includegraphics[width=.4\textwidth]{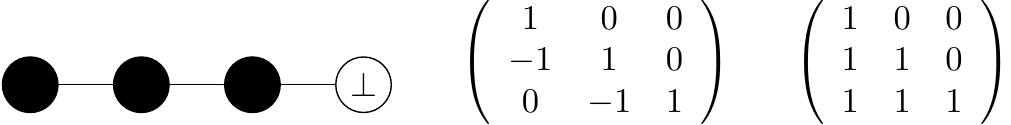}
\caption{\label{fig:graphs} Example policy graph $G$ and their  $\mathbf{P}_G, \mathbf{P}_G^{-1}$}
\end{figure}

It is easy to see that $\mathbf{P}_G$ has all the properties required for our transformational equivalence results to hold. 
\begin{lemma}
  \label{lem:sensitivity-equivalence}
  Let $\mathbf{W}$ be a workload, and $G$ be a policy graph. Then $\Delta_\mathbf{W}(G) = \Delta_{\mathbf{W}_G}$.
\end{lemma}
\eat{
{\sc Proof.}
  This follows from the definition of $\mathbf{P}_G$. We have 
  \begin{align*}
    \Delta_{\mathbf{W}}(G) &= \max_{(\mathbf{x},\mathbf{x}')\in N(G)}\norm{\mathbf{Wx} - \mathbf{Wx}'}_1  \\
    &= \max_{\mathbf{v}_i\in\text{cols}(\mathbf{W}_G)}\norm{\mathbf{v}_i}_1 \rlap{$\qquad \Box$}
  \end{align*} 
}

\begin{lemma}
$\mathbf{P}_G$ constructed above has rank $k$.
\end{lemma}

As $\mathbf{P}_G$ has rank $k$, and $k \leq |E|$, it has a right inverse: $\mathbf{P}_G^{-1} = \mathbf{P}_G^\top(\mathbf{P}_G\mathbf{P}_G^\top)^{-1}$. 
Finally, we prove Claim~\ref{claim:tree-neigh}: When $G$ is a tree $\mathbf{P}_G$ isometrically maps neighbors under policy graph $G$ to neighbors under differential privacy.

\begin{lemma}\label{lem:neighborEquality}
  Suppose $\mathbf{P}_G$ is constructed for a Blowfish policy graph $G$ as above, and $G$ is a tree. Any pair of databases $\mathbf{y}, \mathbf{z} \in \mathbb{R}^k$ are neighbors according to the Blowfish policy $G$ if and only if $\mathbf{P}_G^{-1}\mathbf{y}$ and $\mathbf{P}_G^{-1}\mathbf{z}$ are neighboring databases according to unbounded differential privacy.
\end{lemma}

We refer the reader to \iftoggle{fullpaper}{Appendix~\ref{sec:PG-appendix}}{the full paper} for all the proofs.


\subsubsection*{Case II: Bounded (without $\bot$)}

We next consider a slightly more involved case: let $G=(V,E)$ be a {\em connected} undirected graph, with $V = \dom$, where $|\dom| = k$. If we follow the same construction as in Case I, rows in the resulting $\mathbf{P}_G$ are not linearly independent any more, and thus $\mathbf{P}_G^{-1}$ is not well-defined (no right inverse can be defined for $\mathbf{P}_G$). Fortunately, for every such $G$, we can replace one vertex in $V$ with $\bot$, denoting the resulting graph as $G'$, and correspondingly modify $\mathbf{W}$ and $\mathbf{x}$ to $\mathbf{W}', \mathbf{x}'$ resp., such that (a) $\mathbf{P}_{G'}$ is full rank, (b) answering $\mathbf{W}'$ on $\mathbf{x}'$ under policy $G'$ has the same error as answering $\mathbf{W}$ on $\mathbf{x}$ under $G$, and (c) $\mathbf{W}\mathbf{x}$ can be reconstructed from the answer to $\mathbf{W}'\mathbf{x}'$. 

Pick any value $v \in V$; in $G' = (V', E')$, let $V' = V - \{v\} + \{\bot\}$ and $E' = E - \{(v, u) \mid u \in V\} + \{(\bot, u) \mid (v, u) \in E\}$. Then $G'$ falls into Case I, so we can construct $\mathbf{P}_{G'}$ and $\mathbf{P}^{-1}_{G'}$ as in Case I. 

We  transform $\mathbf{x}$ by removing the entry $\mathbf{x}[v]$ (denoted as $\mathbf{x}_{-v}$). We then transform $\mathbf{W}$ to $\mathbf{W}$ by removing the column $v$ and rewriting all queries that depend on $\mathbf{x}[v]$ to use $n - \sum_{j \neq v} \mathbf{x}[j]$, where $n = \sum_{i \in \dom} \mathbf{x}[i]$ is the size of the input database, without any loss in our ability to answer the original queries. We can do this because when $\bot$ is not in $G$, neighboring databases have the same number of tuples. We can show that our construction satisfies all three requirements (a), (b), and (c) discussed above. 

\begin{lemma}
Consider $G'$, $\mathbf{W}'$, and $\mathbf{x}_{-v}$ constructed above. We have: i) $\mathbf{W} \mathbf{x} = \mathbf{W}' \mathbf{x}_{-v} + \mathbf{c}(\mathbf{W}, n)$, where $\mathbf{c}(\mathbf{W}, n)$ is a constant vector depending only on $\mathbf{W}$ and the size of the database; and ii) any two databases $\mathbf{y}$ and $\mathbf{z}$ are neighbors under $G$ if and only if $\mathbf{y}_{-v}$ and $\mathbf{z}_{-v}$ are neighbors under $G'$.
\label{lem:column-removal}
\end{lemma}

Technical details of the construction and proof of correctness are presented in \iftoggle{fullpaper}{Appendix~\ref{sec:caseII-PG-appendix}}{the full paper}.

\begin{example}\label{ex:fullrank}
  Recall the $\mathbf{C}_k$ workload from Figure~\ref{fig:workloads}.
  In $\mathbf{C}_k$, the last row computes $n$, the size of the database. Since we already know $n$, we do not need to answer that query privately. We can equivalently consider a workload $\mathbf{C'}_k$ with all zeros in the last row and removing the last column (since it would have all zeros). We can also remove the all zero row that remains resulting in a $(k-1)$ by $(k-1)$ matrix. Consider the line graph with $k$ nodes connected in a path. We can replace the rightmost node with $\bot$ (Figure~\ref{fig:graphs}) to get $G'$.   $\mathbf{P}_{G'}$ is a $(k-1) \times (k-1)$ matrix that is full rank, and $\mathbf{P}_{G'}^{-1}$ is equal to $\mathbf{C'}_k$. 

Thus, by Theorem~\ref{thm:blowfishEquality} and Lemma~\ref{lem:column-removal}, the minimum error for answering $\mathbf{C}_k$ under Blowfish policy $G^1_k$ is equal to the minimum error for answering $\mathbf{C'}_k  \cdot \mathbf{P}_{G'} = \mathbf{I}_{k-1}$ under $\epsilon$-differential privacy. Since $\mathbf{I}_{k-1}$ is the identity workload, an optimal data independent strategy would be to add Laplace noise to yield a total error of $\Theta(k/\epsilon^2)$. 
\end{example}

\eat{
\subsubsection*{Case III: General Policy Graph}

A general policy graph may have more than one connected component. If $\bot$ is not connected to a component, we simply apply the conversion discussed in Case II to reduce it to Case I with the vertex $\bot$ connected to it. Then eventually we will have all components connected to $\bot$, which essentially falls into Case I. So Case III can be also reduced to Case I.
We discuss this further in Appendix~\ref{sec:PG-general}.
}

\section{Blowfish Private Mechanisms}
\label{sec:upperbound}
In this section, we derive mechanisms (with near optimal data independent error) for answering range queries and histograms under Blowfish policies to illustrate the power of the transformational equivalence theorem. 
In Section~\ref{sec:workloads-graphs-definition} we define the types of queries and graphs we will be focusing on. 
In Sections~\ref{sec:range-1} and \ref{sec:range-theta} we present strategies for answering multi-dimensional range queries under the grid graph policy. 
These algorithms are data independent and incur the same error on all datasets. In Section~\ref{sec:datadependent}, we extend our strategies to get data dependent algorithms for Blowfish. 
Figure~\ref{fig:summary} summarizes our data independent error bounds. The error incurred by data dependent versions of our algorithms will be evaluated in Section~\ref{sec:experiments}.

\begin{figure}[t]
{\small 
  \centering
  \begin{tabular}{| c | c c | c | }
    \hline
    \multirow{2}{*}{\small \textbf{Workload}} & \multicolumn{3}{|c|}{\textbf{Error per query}} \\
    \cline{2-4}
    & \multicolumn{2}{|c|}{\small \textbf{Blowfish}} & {\small \textbf{$\epsilon$-Diff. \cite{icde:XiaoWG10}}}\\
    \hline
    \multirow{2}{*}{$\mathbf{R}_{k}$} & $G^1_{k}$ & $\Theta(1/\epsilon^2)$ & \multirow{2}{*}{$O(\log^3 k / \epsilon^2)$} \\
    & $G^\theta_{k}$ & $O(\frac{\log^3 \theta}{\epsilon^2})$ &  \\
    \hline
    \multirow{2}{*}{$\mathbf{R}_{k^d}$} & $G^1_{k^d}$ & $O(d \frac{\log^{3(d-1)}k}{\epsilon^2})$ & \multirow{2}{*}{$O(\log^{3d}k / \epsilon^2)$} \\
    & $G^\theta_{k^d}$ & $O(d^3 \frac{ \log^{3(d-1)}k \log^3{\theta}}{\epsilon^2})$ &  \\
    \hline
  \end{tabular}
  \caption{\label{fig:summary} Summary of data independent error bounds.}
}\end{figure}
\subsection{Workloads and Policy Graphs}\label{sec:workloads-graphs-definition}

We note that the data independent mechanisms we present for one dimensional range queries under $G^1_k$ and $G^\theta_k$ (Sections~\ref{sec:range-1Dline} and \ref{sec:range-1Dtheta}) are similar to the ones presented in the original Blowfish paper \cite{blowfish}. We present them here to illustrate our transformational equivalence and subgraph approximation results, and to help the reader understand our novel mechanisms for multi-dimensional range queries. 

Consider a multidimensional domain $\dom = [k]^d$, where $[k]$ denotes the set of integers between $1$ and $k$ (inclusive). The size of each dimension is $k$ and thus the domain size is $k^d$. A database in this domain can be represented as a (column) vector ${\bf x} \in \mathbb{R}^{k^d}$ with each entry ${\bf x}_i$ denoting the true count of a value $i \in \dom$. It is important to note that our results in this paper can be easily extended to the case when dimensions have different sizes.

A multidimensional range query can be represented as a $d$-dimensional hypercube with the bottom left corner $\mathbf{l}$ and the top right corner $\mathbf{r}$. In particular, when $d=1$, a range query ${\bf q}({\bf l}, {\bf r})$ is a linear counting query which count the values within $\bf l$ and $\bf r$ in the database $\bf x$, i.e., ${\bf q}({\bf l}, {\bf r}) {\bf x} = \sum_{{\bf l} \leq i \leq {\bf r}} {\bf x}_i$. Let $\mathbf{R}_k$ denote the workload of all such one dimensional range queries, \i.e., $\mathbf{R}_k = \{{\bf q}({\bf l}, {\bf r}) \mid {\bf l}, {\bf r} \in [k] \wedge {\bf l} \leq {\bf r}\}$. Similarly, let $\mathbf{R}_{k^d} = \{{\bf q}({\bf l}, {\bf r}) \mid {\bf l}, {\bf r} \in [k]^d \wedge {\bf l} \leq {\bf r}\}$ denote the workload of all $d$-dimensional range queries. Note that each range query can be represented as a $k^d$-dimensional row vector, and $\mathbf{R}_{k^d}$ can be represented as a $q \times k^d$ matrix, where $q= (k(k-1)/2)^d$ is the total number of range queries.


The class of policy graphs $G^{\theta}_{k^d} = (V, E)$ we consider here are called {\em distance-threshold} policy graphs. They are defined based on the $L_1$ distance in the domain $\dom = [k]^d$. Consider two vertices ${u} = (u_1,\dots,u_d)$ and ${v} = (v_1,\dots,v_d) \in V \subseteq [k]^d$, the $L_1$ distance between is $|{u}-{v}| = |u_1-v_1| + \dots + |u_d - v_d|$. There is an edge $({u},{v})$ in $E$ if and only if $|{u}-{v}| \leq \theta$.
Two special cases of $G^{\theta}_{k^d}$ and their semantics were discussed in \csec\ref{sec:blowfish} as line graph ($G^1_k$) and grid graph ($G^\theta_{k^2}$).

%
%

\subsection{Range Queries under $G^1_{k^d}$}\label{sec:range-1}

In this section we first describe the easy case of 1D range queries before considering multi-dimensional range queries. We will heavily utilize the structure of the transformed query workload in this section. The following lemma helps relate the queries in $\mathbf{W}$ to the queries in $\mathbf{W}_G$.

\begin{lemma}\label{lem:query-transformation}
  Let $\mathbf{q}$ be a linear counting query (that is, all entries in $\mathbf{q}$ are either $1$ or $0$), and $G=(V,E)$ be a policy graph. Let $\{v_1,\dots,v_\ell\} \subseteq V$ be the vertices corresponding to the nonzero entries of $\mathbf{q}$. Then, the nonzero columns of $\mathbf{q} \cdot \mathbf{P}_G = \mathbf{q}_G$ correspond to the set of edges $(u,v)$ with exactly one end point in $\{v_1,\dots,v_\ell\}$. That is, 
\begin{equation*}
  \left\{ (u,v) |\left\{ u,v \right\} \cap \left\{ v_1,\dots,v_\ell \right\}| = 1  \right\}.
\end{equation*}
\end{lemma}
\begin{proof}
  Each entry $c$ of $\mathbf{q}_G$ satisfies $c = u - v$ where $u,v$ are entries in $\mathbf{q}$ and $(u,v) \in E$. $c$ is nonzero exactly when $u \ne v$, or equivalently, when
  \begin{equation*}
    |\left\{ u,v \right\} \cap \left\{ v_1,\dots,v_k \right\}| = 1. \qed
  \end{equation*}
\end{proof}

\subsubsection{$\mathbf{R}_k$ under $G^1_k$}\label{sec:range-1Dline}
\begin{algorithm}[t]
  \caption{1D range queries.}
  \label{alg:range-1Dline}
{\small
  \begin{algorithmic}[1]
    \Require
      \Statex $\mathbf{W}$ is a workload of range queries, $\mathbf{x}$ is a database.
    \Function{1DRange}{$\mathbf{W},\mathbf{x}$}
      \Let{$\mathbf{x}_{G}$}{$\mathbf{P}_{G^1_k}^{-1}\mathbf{x}$  {\small {\em // prefix sums from $\mathbf{x}$}}}
      \Let{$\tilde{\mathbf{x}}_{G}$}{Differentially private estimate for $\mathbf{x}_{G}$}
      \Let{$\mathbf{W}_{G}$}{$\mathbf{WP}_{G^1_k}$ {\small {\em // differences between prefix sum pairs}} }
      \State \Return $\mathbf{W}_{G} \tilde{\mathbf{x}}_{G^1_k}$
    \EndFunction
  \end{algorithmic}
}
\end{algorithm}
We begin with a simple case: one-dimensional range queries under a one-dimensional line graph. We outline the application of Theorem~\ref{thm:blowfishEquality} in Algorithm~\ref{alg:range-1Dline}. Recall from Example~\ref{ex:fullrank} that the inverse of $\mathbf{P}_{G^1_k}$ is the  cumulative histogram workload. Therefore, the transformed database $\mathbf{x}_G = \mathbf{P}_{G^1_k}^{-1}\mathbf{x}$ corresponds to the set of prefix sums in $\mathbf{x}$. Algorithm~\ref{alg:range-1Dline} computes a differentially private estimate of $\mathbf{x}_G$ (say using the Laplace mechanism).

\begin{figure}[h]
  \centering
  \includegraphics[width=.2\textwidth]{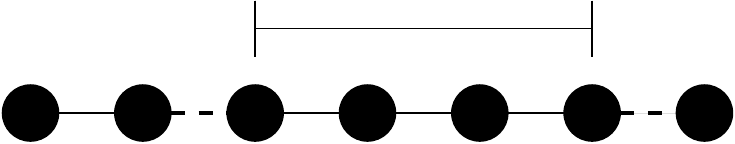}
  \caption{A one dimensional range query on vertices is transformed into a query on edges (represented by dashed lines).}
  \label{fig:1Dline}
\end{figure}

Next, we transform the queries. Note that for any range query $\mathbf{q} = [l,r]$, by Lemma~\ref{lem:query-transformation} $\mathbf{q}_G$ contains at most two nonzero elements (this is illustrated in Figure~\ref{fig:1Dline}) corresponding to the edges $(l-1, l)$ and $(r, r+1)$. The values of $\mathbf{x}_G$ at these edges are the prefix sums $\sum_{i=1}^{l-1} x_i$ and $\sum_{i=1}^r x_i$, and their difference is indeed the answer to the original range query. 
We can show the following bound on the data independent error of Algorithm~\ref{alg:range-1Dline}.
\begin{theorem}\label{thm:range-1Dline}
  Algorithm~\ref{alg:range-1Dline} with $\tilde{\mathbf{x}}_{G_k^1} = \mathbf{x}_{G_k^1} + Lap(1/\epsilon$) answers workload $\mathbf{R}_k$ with $\Theta(1/\epsilon^2)$ error per query under $(\epsilon, G^1_k)$-Blowfish privacy.
\end{theorem}
\begin{proof}
  Every $\mathbf{q}_{G^1_k}(l, r) \in \mathbf{R}_{G^1_k}$ can be reconstructed by summing at most two queries in $\tilde{\mathbf{x}_{G_k^1}}$.
  Each entry in $\tilde{\mathbf{x}_{G_k^1}}$ has $\Theta(1/\epsilon^2)$ error from the Laplace mechanism.
  So each $\mathbf{q}_{G^1_k}(l, r)$ incurs only $\Theta(1/\epsilon^2)$ error.
\end{proof}
In fact we show that Algorithm~\ref{alg:range-1Dline} is the optimal data independent algorithm for answering range queries in 1D under $G^1_k$. We omit the proof due to space constraints.
\begin{lemma}
  Any $(\epsilon, G^1_k)$-Blowfish private mechanism answers $\mathbf{R}_k$ with $\Omega(1/\epsilon^2)$ error per query.
  \label{lem:range-1Dline-lower-bound}
\end{lemma}

The best known data independent strategy (with minimum error) for answering $\mathbf{R}_k$ under $\epsilon$-differential privacy is the Privelet strategy \cite{icde:XiaoWG10} with a much larger asymptotic error of $O(\log^3 k / \epsilon^2)$ per query.

\subsubsection{$\mathbf{R}_{k^d}$ under $G^1_{k^d}$}\label{sec:range-2Dline}

\begin{figure}[t]
  \captionsetup[subfigure]{width=\textwidth}
  \centering
  \begin{subfigure}[t]{.22\textwidth}
    \centering
    \includegraphics[width=.6\textwidth]{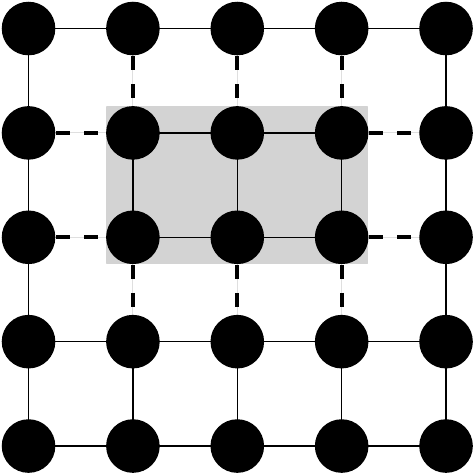}
    \caption{$G^1_{5^2}$ with a two dimensional range query, represented by a grey box. The edges in the transformed query (satisfying Lemma~\ref{lem:query-transformation}), are shown with dashed lines. These edges form four ranges.}
    \label{fig:2Dline-b}
  \end{subfigure}
  \quad
  \begin{subfigure}[t]{.22\textwidth}
    \centering
    \includegraphics[width=.6\textwidth]{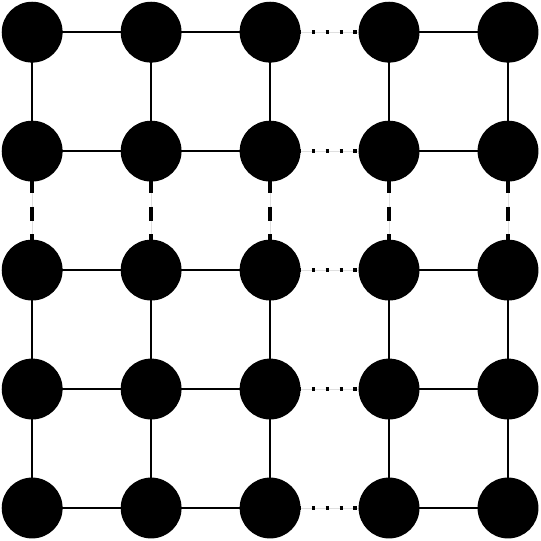}
    \caption{For each row of vertical edges, we answer all ranges over the row. One such row is shown with dashed lines. We must do the same for columns, and one such column is shown with dotted lines.}
    \label{fig:2Dline-c}
  \end{subfigure}
  \caption{Answering $\mathbf{R}_{k^2}$ under $G^1_{k^2}$.}
  \label{fig:2Dline}
\end{figure}

$G^1_{k^d}$ is a grid with $k^d$ vertices and $d\cdot (k-1)\cdot k^{d-1}$ edges. Let us consider the problem in two dimensions first. (see Figure~\ref{fig:2Dline-b}).  The transformed domain (after using Theorem~\ref{thm:blowfishEquality}) would be the set of edges in  the graph. Consider a 2D range query $\mathbf{q}([x,y], [x',y'])$ (grey box in figure). The transformed query $\mathbf{q}_G$ has non-zero entries corresponding to edges on the boundary of the original range query (dashed lines in the figure). Note that these edges can be divided into 4 contiguous ranges of edges; i.e., $\mathbf{q}_G$ is the sum of 4 disjoint range queries in the transformed domain.

\eat{
\begin{algorithm}[t]
  \caption{Multi-D range queries.}
  \label{alg:range-multi-line}
  \begin{algorithmic}[1]
    \Require
      \Statex $\mathbf{W}$ is a workload of range queries, $\mathbf{x}$ is a database.
    \Function{MultiDRange}{$\mathbf{W},\mathbf{x}$}
      \Let{$\mathbf{x}_{G}$}{$\mathbf{P}_{G^1_{k^d}}^{-1}\mathbf{x}$  {\small {\em // dataset on edges of $G^1_{k^d}$}}}
      \Let{$\mathbf{W}_{G}$}{$\mathbf{WP}_{G^1_k}$  {\small {\em // transformed queries on edges of $G^1_{k^d}$}}}
      \For{$i \in [1,d]$ and $j \in [1,k-1]$}
        \Let{$\mathbf{A}^{i,j}_G$}{ $\{(\mathbf{l},\mathbf{m}) \ | \ \mathbf{l}[i] = j$, $\mathbf{m}[i] = j+1\}$}
	 \State Answer all range queries in $\mathbf{A}^{i,j}_{G}$ using $\epsilon$-differential privacy
      \EndFor
      \State Each queries in $\mathbf{W}_{G}$ is a sum of $2d$ answers from $\{\mathbf{A}^{i,j}_{G}\}_{i,j}$
    \EndFunction
  \end{algorithmic}
\end{algorithm}}
Thus, a strategy for answering the transformed query workload in two dimensions would be to answer all one dimensional range queries along the rows (dashed vertical edges in Fig~\ref{fig:2Dline-c}) and columns (dotted horizontal edges in Fig~\ref{fig:2Dline-c}) under differential privacy. There are $2(k-1)$ such sets of range queries. Note that these sets of range queries are disjoint, and can each be answered using $\epsilon$-differential privacy (under parallel composition).  Any query $\mathbf{q}_G$ can be computed by adding up the answers to 2 row range queries and 2 column range queries.

 In $d$ dimensions, $\mathbf{q}_G$ will  be the  sum of $2d$ $(d-1)$-dimensional range queries on the transformed dataset $\mathbf{x}_G$, each corresponding to a face of the $d$-dimensional range query. Our strategy would be to answer  $d(k-1)$  sets of $(d-1)$-dimensional range queries under $\epsilon$-differential privacy. 


We bound the data independent error of our algorithm.
\begin{theorem}\label{thm:range-dd}
  Workload $\mathbf{R}_{k^d}$ can be answered with 
  \begin{equation*}
    O({d \log^{3(d-1)}k}/{\epsilon^2})
  \end{equation*}
  error per query under $(\epsilon, G^1_{k^d})$-Blowfish privacy.
\end{theorem}
\begin{proof}
For each dimension, we must answer $k-1$ sets of $(d-1)$-dimensional range queries, for a total of $(k-1) \cdot d$ sets of $(d-1)$-dimensional ranges. As we have shown, all of these sets are disjoint and can be answered in parallel. Therefore, the total error is just the error of answering one of these sets of ranges. We can answer these ranges using Privelet \cite{icde:XiaoWG10} with $O(\frac{\log^{3(d-1)}k}{\epsilon^2})$ error. To answer our query, we must sum $2d$ of these ranges for a total error of
  $O(d\frac{ \log^{3(d-1)}k}{\epsilon^2})$.

By Theorem~\ref{thm:blowfishEquality}, we can answer $\mathbf{R}_{k^d}$ under $G^1_{k^d}$ with the same error per query.
\end{proof}

We get a $\Omega(\log^3{k})$ factor better error than differential privacy using Privelet \cite{icde:XiaoWG10} under a fixed dimensionality $d$.

\subsection{Range Queries under $G^\theta_{k^d}$}\label{sec:range-theta}
We next consider answering range queries under a more complex graph. Unlike in the case of $G^1_{k^d}$, the workloads resulting from the use of Theorem~\ref{thm:blowfishEquality} to the $G^\theta_{k^d}$ policy are not well studied under differential privacy. Hence, we will introduce a new tool, called {\em subgraph approximation}, and then use it to design Blowfish private mechanisms.

\subsubsection{$\mathbf{R}_k$ under $G^\theta_k$}\label{sec:range-1Dtheta}
We next present an algorithm for answering one dimensional range queries under, $G^\theta_k$. These results generalize the results from \csec\ref{sec:range-1Dline}, and will leverage subgraph approximation (Lemma~\ref{lem:subgraph}).

\begin{figure*}[ht!]
  \captionsetup[subfigure]{width=.95\textwidth}
  \centering
  \begin{subfigure}[t]{.49\textwidth}
    \centering
    \includegraphics[width=.7\textwidth]{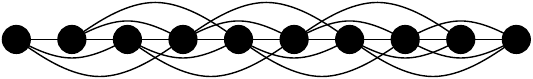}
    \caption{$G^3_{10}$, each vertex is connected to other vertices within distance 3 along the line.}
    \label{fig:1Dtheta-a}
  \end{subfigure}
  \begin{subfigure}[t]{.49\textwidth}
    \centering
    \includegraphics[width=.7\textwidth]{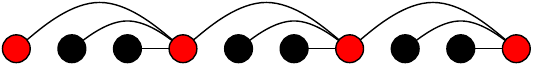}
    \caption{$H^3_{10}$, each vertex is connected to the nearest red vertex to its right.} 
    \label{fig:1Dtheta-b}
  \end{subfigure}
  \begin{subfigure}[t]{.49\textwidth}
    \centering
    \includegraphics[width=.7\textwidth]{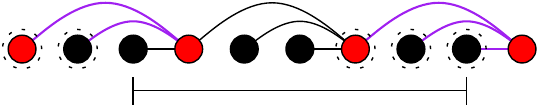}
  \caption{A range query shown on ${H}^3_{10}$. The transformed query consists of the edges highlighted in purple and their left end points always form two contiguous ranges.}
    \label{fig:1Dtheta-c}
  \end{subfigure}
  \begin{subfigure}[t]{.49\textwidth}
    \centering
    \includegraphics[width=.7\textwidth]{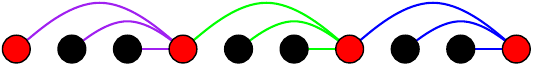}
    \caption{Our strategy will answer all range queries on 3 sets of edges, each set shown in a different color. These sets of edges are disjoint.}
    \label{fig:1Dtheta-d}
  \end{subfigure}
  \caption{A summary of a strategy for answering $\mathbf{R}^k$ under $G^\theta_k$ for $\theta=3,k=10$. Our results hold in general.}
  \label{fig:1Dtheta}
\end{figure*}

We first describe how to obtain a subgraph $H^\theta_{k}$ from $G^\theta_{k}$. We designate $k/\theta$ vertices at intervals of $\theta$; call these ``red'' vertices. In $H^\theta_{k}$, consecutive red vertices are connected to form a path (like the line graph). All non-red vertices are only connected to the next red vertex (to its right); i.e., vertices $\{1, 2, \ldots, \theta-1\}$ are connected only to vertex $\theta$, vertices $\{\theta+1, \theta+2, \ldots, 2\theta-1\}$ are connected only to vertex $2\theta$, and so on. We order the edges in $H^\theta_k$ by their left endpoints.

Like $G^1_k$, $H^\theta_k$ is also a tree with $k-1$ edges.  Figure~\ref{fig:1Dtheta-a} shows $G^3_{10}$ and Figure~\ref{fig:1Dtheta-b} shows $H^3_{10}$. Note that for all $\theta$, a pair of adjacent vertices in $G^\theta_k$ are connected by a path of length $\ell \leq 3$ in $H^\theta_k$. So we can use subgraph approximation.

\eat{
\begin{algorithm}
  \caption{1D range queries under theta graph.}
  \label{alg:range-1Dtheta}
  \begin{algorithmic}[1]
    \Require
      \Statex $\mathbf{W}$ is a workload of range queries, $\mathbf{x}$ is a database.
      \Statex $\mathbf{R}_{\left[ ik,ik+\theta \right]}$ is the set of range queries in $\mathbf{R}_{k-1}$ with left and right ends of the range falling in $\left[ i\theta,i\theta+\theta \right]$.
    \Function{1DRangeTheta}{$\mathbf{W},\mathbf{x}$}
      \Let{$\mathbf{W}_{H^1_k}$}{$\mathbf{WP}_{H^1_k}$}
      \Let{$\mathbf{x}_{H^1_k}$}{$\mathbf{P}_{H^1_k}^{-1}\mathbf{x}$}
      \For{$i \gets 0 \text{ to } \theta/k -1$}
      \State Add all rows of $\mathbf{R}_{\left[ i\theta,i\theta+\theta \right]}$ to $\mathbf{U}$.
      \EndFor
      \Let{$\tilde{\mathbf{U}\mathbf{x}_{H^1_k}}$}{Noisy estimate for $\mathbf{U}\mathbf{x}_{H^1_k}$}
      \For{$\mathbf{q} \in \mathbf{W}_{H^1_k}$}
      \State Compute $\tilde{\mathbf{qx}}$ with the sum of 2 rows of $\tilde{\mathbf{U}\mathbf{x}_{H^1_k}}$.
        \State Add $\tilde{\mathbf{qx}}$ as a row in $\tilde{\mathbf{W}_{H^1_k}\mathbf{x}_{H^1_k}}$
      \EndFor
      \State \Return $\tilde{\mathbf{W}_{H^1_k} \mathbf{x}_{H^1_k}}$
    \EndFunction
  \end{algorithmic}
\end{algorithm}} 

  Consider some query in $\mathbf{R}_k$, say $\mathbf{q}(l,r)$. The corresponding query $\mathbf{q}_{H^\theta_k}$ in $\mathbf{R}_{H^\theta_{k}}$ consists of all edges with one of the end points within the range $(l,r)$ (Lemma~\ref{lem:query-transformation}). If $l \leq x\theta \leq r \leq y\theta$, where $x\theta$ and $y\theta$ are the smallest red nodes greater than $l$ and $r$, then these edges correspond to $\{(i, x\theta) \mid (x-1)\theta \leq i < l\}$ and $\{(j, y\theta) \mid (y-1)\theta \leq j < r\}$ (edges connected to dotted nodes in Figure~\ref{fig:1Dtheta-c}). That is, the transformed query $\mathbf{q}_{H^\theta_k}(l,r)$ corresponds to the difference of two range queries (according to the ordering of edges in $H^\theta_k$). Moreover, each range query is of length at most $\theta$ -- within $[(x-1)\theta, x\theta]$ for some $x$.
   
   Thus, our strategy for answering all the queries in $R_{H^\theta_k} = R_k\cdot P_{H^\theta_k}$ is as follows.  Partition the transformed domain (or edges in $H^\theta_k$) into disjoint groups of $\theta$ -- all edges connecting a red node to nodes on its left form a group (see Figure~\ref{fig:1Dtheta-d}). Next, answer all range queries of length at most $\theta$ within each of these groups under $\epsilon$-differential privacy (say using Privelet).  Finally, reconstructing queries $\mathbf{q}_{H^\theta_k}(l,r) \in R_{H^\theta_k}$ using the computed range queries. Since these sets of range queries form disjoint subsets of the domain, they all can use the same $\epsilon$ privacy budget (by parallel composition). 

\begin{theorem}
There exists a mechanism that answers workload $\mathbf{R}_k$ with
  \begin{equation*}
    O\left({\log^3 \theta}/{\epsilon^2}\right)
  \end{equation*}
  error per query under $(\epsilon,G^\theta_{k})$-Blowfish privacy.
  \label{thm:onedimrange-theta}
\end{theorem}
\begin{proof}
Our strategy partitions the transformed domain (or edges in $H^\theta_k$) into groups of size $\theta$, and answers range queries over them. Using privelet to compute range queries within each partition results in $O(\log^3 \theta/\epsilon^2)$ error per query. Queries in 
$R_{H^\theta_k}$ are differences of at most 2 range queries, and thus also incur at most $O(\log^3 \theta/\epsilon^2)$ error. By Theorem~\ref{thm:blowfishEquality}, the same mechanism answers $\mathbf{R}_k$ and satisfies $(\epsilon, H^\theta_k)$-Blowfish privacy with $O(\log^3 \theta/\epsilon^2)$ error per query.

For every edge $(u,v) \in G^\theta_k$, $u$ and $v$ are connected by a path of length at most $3$, this strategy also ensure $(3\epsilon, G^\theta_k)$-Blowfish privacy. Thus, using the above strategy with privacy budget $\epsilon/3$ gives us the required result.
\end{proof}

\subsubsection{$\mathbf{R}_{k^d}$ under $G^\theta_{k^d}$}

\begin{figure*}[ht!]
  \captionsetup[subfigure]{width=\textwidth}
  \centering
  \begin{subfigure}[t]{.2\textwidth}
    \centering
    \includegraphics[width=.8\textwidth]{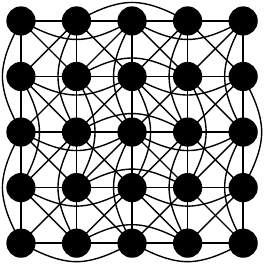}
    \caption{}
    \label{fig:2Dtheta-a}
  \end{subfigure}
  \begin{subfigure}[t]{.2\textwidth}
    \centering
    \includegraphics[width=.7\textwidth]{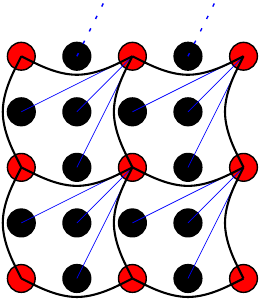}
    \caption{}
    \label{fig:2Dtheta-b}
  \end{subfigure}
  \begin{subfigure}[t]{.2\textwidth}
    \centering
    \includegraphics[width=.8\textwidth]{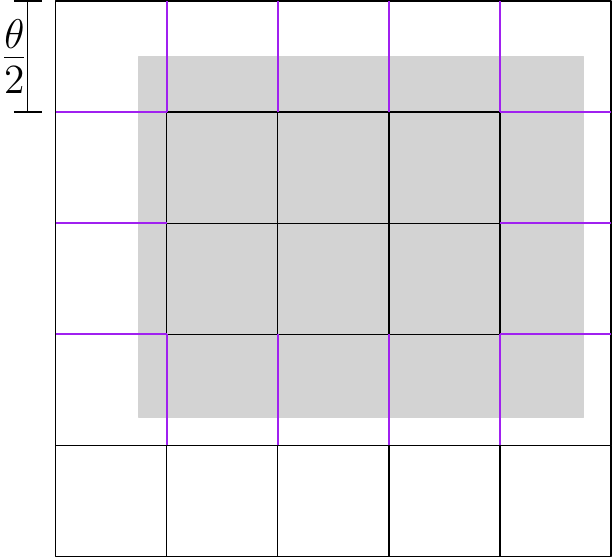}
    \caption{}
    \label{fig:2Dtheta-c}
  \end{subfigure}
  \begin{subfigure}[t]{.2\textwidth}
    \centering
    \includegraphics[width=.8\textwidth]{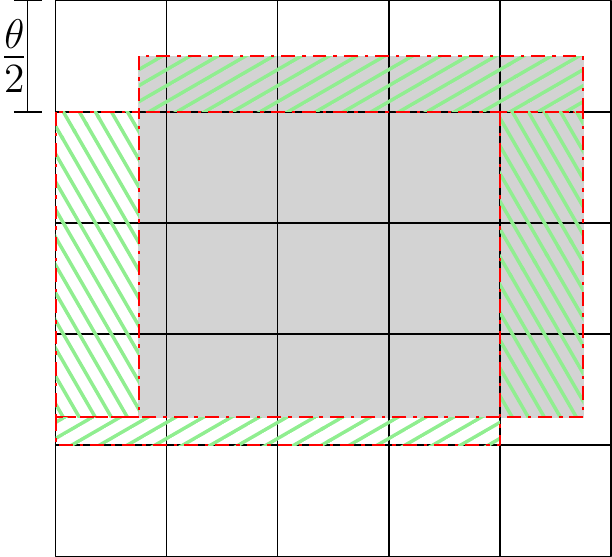}
    \caption{}
    \label{fig:2Dtheta-d}
  \end{subfigure}
  \caption{\label{fig:2Dtheta} Transforming queries in $\mathbf{R}_{k^2}$ under $G^\theta_{k^2}$. (a) $G^2_{5^2}$, edges are  vertices with $L_1$ distance $\leq 2$. (b) A section of $H^2_{k^2}$. Internal edges connect red nodes and external edges connect black nodes to red nodes. (c) A 2D range query superimposed on $H^2_{k^2}$. We only show the divisions in $\theta/2$ blocks of vertices. Within each block, all vertices would be connected to the upper right corner. The grid of lines shows all the external edges. Highlighted in purple are the external edges which satisfy Lemma~\ref{lem:query-transformation} and therefore appear in the transformed query. (d) The green patterned rectangles show the sets of vertices corresponding to the internal edges in the transformed query. There are 4 such rectangles, and for each one either the height or length is bounded by $\theta$. Note that there are other ways in which we could divide the green patterned region into 4 rectangles, we arbitrarily chose one. Our strategy is to answer all range queries over each row of squares, and each column of squares. We illustrate using specific values for $\theta$ and $k$, but our results hold in general.}
\end{figure*}

We now turn our attention to multidimensional range queries under  $G^\theta_{k^d}$. Our strategy will be similar to the one in Section~\ref{sec:range-1Dtheta}. We find a subgraph to approximate $G^\theta_{k^d}$. We show the queries of the transformed workload can be decomposed into range queries on edges of bounded size, and our strategy is to answer these range queries. 

To get a subgraph $H^\theta_k$, we divide $G^\theta_{k^d}$ into $d$-dimensional hypercubes with edge length $\frac{\theta}{d}$ (see Figures~\ref{fig:2Dtheta-a} and \ref{fig:2Dtheta-b}). We designate the vertices at the corners of the cubes as ``red'' vertices. We pick a mapping of hypercubes to red vertices. For example, in the 2-dimensional case, we may map each square to its upper right red vertex. Each non-red vertex within a cube is connected to this selected red vertex (we pick a consistent mapping for vertices that are on the boundary of cubes). We call these {\em internal} edges. The red vertices are then connected in a grid (like $G^1_{k^d}$) with {\em external} edges.

Due to space constraints we give a brief sketch of our strategy. Given a range query $\mathbf{q}$, the transformed query will correspond to a set of internal and a set of external edges (as per Lemma~\ref{lem:query-transformation}). Since external and internal edges are disjoint, our strategy answers the transformed query restricted to the external edges and internal edges independently, each under $\epsilon$-differential privacy. 

Since the external edges form a grid graph $G^1_{(dk/\theta)^d}$ (Figure~\ref{fig:2Dtheta-c}), we can use our strategy from \csec\ref{sec:range-2Dline} to answer this part of the transformed query with error at most $O(d \frac{\log^{3(d-1)} {d \cdot k}/\theta}{\epsilon^2})$. We can show that the non-red end points of the internal edges featuring in the transformed query for $2d$ $d$-dimensional range queries. However, each of these range queries has a width at most $\theta$ in one of the dimensions (see Figure~\ref{fig:2Dtheta-d}). Thus like in \csec\ref{sec:range-2Dline}, it is sufficient to answer all $d$-dimensional range queries having width at most $\theta$ in one dimension. However, the edges in these range queries are not disjoint, and hence we can only use $\theta/d$ privacy budget for answering these range queries, resulting in the following error bound:

\begin{theorem}
  Workload $\mathbf{R}_{k^d}$ can be answered with
  \begin{equation*}
    O(d^3 \cdot \frac{\log^{3(d-1)}{k}\log^3{\theta}}{\epsilon^2})
  \end{equation*}
  error per query under $(\epsilon,G^\theta_{k^d})$-Blowfish privacy.
  \label{thm:twodimrange-theta}
\end{theorem}

\stitle{Discussion.} Our Blowfish mechanisms under $G^1_{k^d}$ and $G^\theta_{k^d}$ policy graphs improve upon Privelet by a factor of $\log^3 k$, but incur additional error by a factor of $d$ and $d^3\log^3\theta$ respectively. Thus the proposed mechanisms are better than using Privelet when $d \log \theta$ is small compared to $\log k$. This is true in the case of location privacy where $d = 2$ and $\theta$ (10s of km) is usually much smaller than $k$ (1000s of km).



\subsection{Data Dependent Algorithms}
\label{sec:datadependent}
Till now we considered data independent Blowfish mechanisms whose error is independent of the input database. Recent work has investigated a new class of data dependent algorithms for answering histogram and range queries that exploit the properties of the data and incur much lower error on some (typically sparse) datasets. In this section, we present two methods for adapting the previously given mechanisms to get data dependent mechanisms for Blowfish. 

\subsubsection{Data dependent differentially private algorithms}
In all of our algorithms we employed data independent differentially private algorithms in our strategies. We used the Laplace mechanism for computing noisy histograms in our strategy for answering $\mathbf{R}_k$ under $G^1_k$, and used Privelet for answering noisy range queries in all the other cases. Instead, when the policy graph is a tree, we could use a state-of-the-art data dependent technique like DAWA \cite{li14:dawa} for answering histograms and range queries under differential privacy. For instance, DAWA computes a noisy histogram as follows: (a) partition the domain such that domain values within a group have roughly the same counts, (b) estimates the total counts for each of these groups using the Laplace mechanism, and (c) uniformly divides the noisy group totals amongst its constituents. When many counts are similar (especially when $\mathbf{x}$ is sparse), DAWA incurs lower error than Laplace mechanism since it adds noise to fewer counts (see \csec\ref{sec:experiments}). 

\subsubsection{Using properties of the transformed database}
All the example workload/policy pairs discussed in this section are such that the transformed workload $\mathbf{W}_G$ is ``easier'' to answer under differential privacy than $\mathbf{W}$. Thus,  the data independent Blowfish algorithms outperform the data independent differentially private algorithms for answering $\mathbf{W}$ on $\mathbf{x}$ by more than a constant factor. However, this is not true for all workloads. 

Consider, for instance, the identity workload $\mathbf{I}_k$ that computes the histogram of counts under the line graph $G^1_k$. The transformed workload $\mathbf{I}_{G^1_k}$ is the set of differences between adjacent elements in the transformed database; i.e., $\mathbf{x}_G[i] - \mathbf{x}_G[i-1]$, for all $i$. The transformed workload seems no easier than the original workload. 

However, we can utilize the fact that $\mathbf{x}_G$ has special structure. Recall that $\mathbf{P}_{G^1_k}^{-1}$ is precisely equal to the cumulative histogram workload $\mathbf{C}_k$. Thus, the counts in the transformed database $\mathbf{x}_G$ are prefix sums of the counts in $\mathbf{x}$, and are non-decreasing. We can use this property of $\mathbf{x}_G$ to reduce error by enforcing the non-decreasing constraint on the noisy counts $\tilde{\mathbf{x}}_G$. Hay et al \cite{vldb:HayRMS10} present a simple algorithm to postprocess the counts so that the error depends on the number of distinct counts in $\mathbf{x}_G$. Note that whenever a count in $\mathbf{x}$ is $0$, a pair of consecutive prefix sums in $\mathbf{x}_G$ are the same. Therefore, the number of distinct values in $\mathbf{x}_G$ is precisely the number of non-zero entries in $\mathbf{x}$. This suggests that postprocessing $\tilde{\mathbf{x}}_G$ to ensure that the counts are non-decreasing will lead to a significant reduction of error for sparse datasets. We can use this strategy whenever $\mathbf{P}_{G}^{-1}$ creates constraints in $\mathbf{x}_G$. 

\section{Experiments}
\label{sec:experiments}
\newcommand{\range}{{\sc 1D-Range}\xspace}
\newcommand{\tworange}{{\sc 2D-Range}\xspace}
\newcommand{\hist}{{\sc Hist}\xspace}
\newcommand{\marg}{{\sc Marg}\xspace}

\begin{table*}[t]
\centering
{\small
\begin{tabular}{|c|l|r|r|r|}
\hline
& Description & Domain & Scale & \% Zero\\
&& Size & & Counts\\
\hline
\hline
A&Histogram of new links by time added to a subset of the US patent citation network &  4096 & $2.8\times 10^7$ & 6.20 \\
B&Histogram of personal income from 2001-2011 American community survey &  4096  & $2.0\times 10^7$ & 44.97 \\
C&Histogram of new links by time added to HepPH citation network &  4096 & $3.5\times 10^5$ & 21.17 \\
D&Frequency of search term ``Obama'' over time (2004-2010) & 4096 & $3.4\times 10^5$ & 51.03 \\
E&Number of external connections made by each internal host in an & &&\\
&IP-level network trace collected
at the gateway router of a major university. & 4096 & $2.6\times 10^4$ & 96.61 \\
F&Histogram on ``capital loss'' attribute of Adult US Census dataset& 4096  & $1.8\times 10^4$ & 97.08 \\
G&Histogram of personal medical expenses based on &&&\\
& a national home and hospice care survey from 2007 & 4096 & $9.4 \times 10^3$ & 74.80 \\
\hline
T100 & Aggregated counts of number of tweets by geo location & $100 \times 100$ & $1.9\times 10^5$ & 84.93 \\
T50 & collected over 24 hours restricted to a bounding box of & $50 \times 50$ & $1.9\times 10^5$ & 69.24 \\
T25 & 50N,125W and 30N,110W (western USA) & $25 \times 25$ & $1.9\times 10^5$ & 43.20 \\
\hline
\end{tabular}
}
\caption{\label{tab:datasets} Description of datasets}
\end{table*}

\begin{figure*}[th!]
  \captionsetup[subfigure]{width=\textwidth}
  \centering
  \begin{subfigure}[t]{.23\textwidth}
    \centering
    \includegraphics[width=\textwidth]{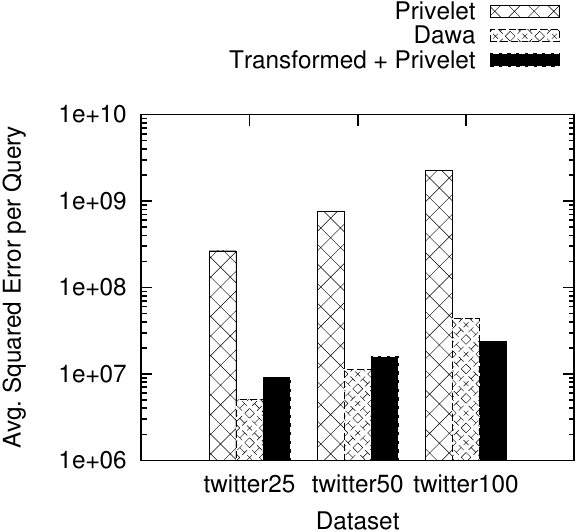}
    \caption{\label{fig:2Drange_.01} \tworange $(\epsilon = 0.01, G^1_{k^2})$}  
  \end{subfigure}
  \begin{subfigure}[t]{.23\textwidth}
    \centering
    \includegraphics[width=\textwidth]{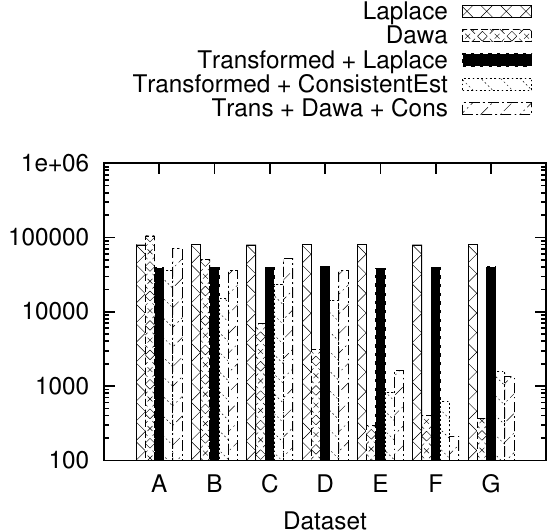}
    \caption{\label{fig:hist_.01} \hist $(\epsilon = 0.01, G^1_k)$}
  \end{subfigure}
  \begin{subfigure}[t]{.23\textwidth}
    \centering
    \includegraphics[width=\textwidth]{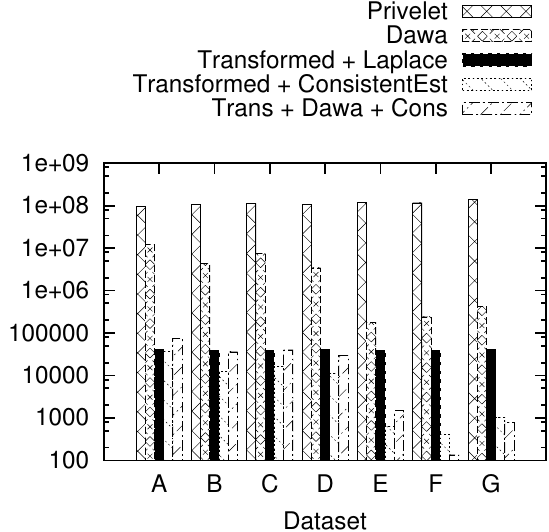}
    \caption{\label{fig:1Drange_.01} \range  $(\epsilon = 0.01, G^1_k)$}  
  \end{subfigure}
  \begin{subfigure}[t]{.23\textwidth}
    \centering
    \includegraphics[width=\textwidth]{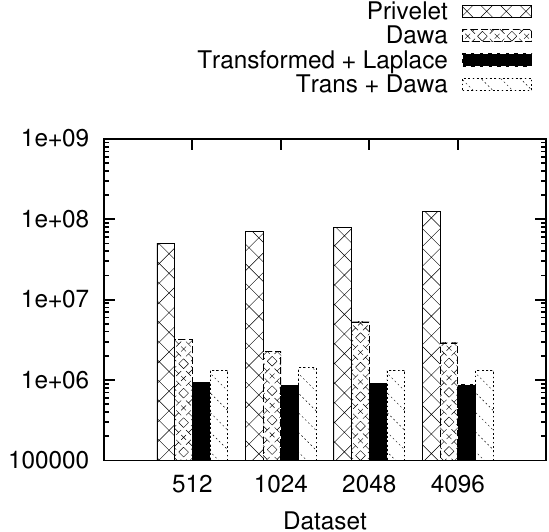}
    \caption{\label{fig:theta_.01} \range  $(\epsilon = 0.01, G^4_k)$}  
  \end{subfigure}

  \begin{subfigure}[t]{.23\textwidth}
    \centering
    \includegraphics[width=\textwidth]{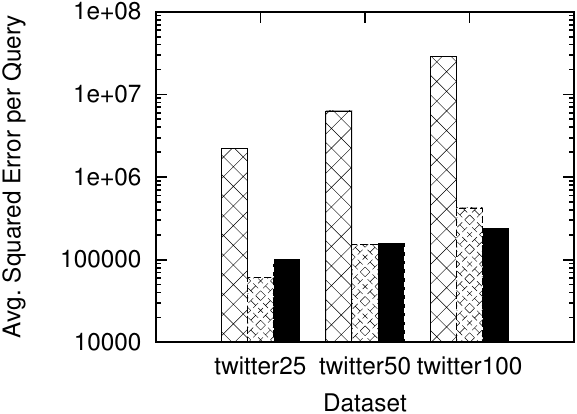}
    \caption{\label{fig:2Drange_0.1}  \tworange  $(\epsilon = 0.1, G^1_{k^2})$}  
  \end{subfigure}
  \begin{subfigure}[t]{.23\textwidth}
    \centering
    \includegraphics[width=\textwidth]{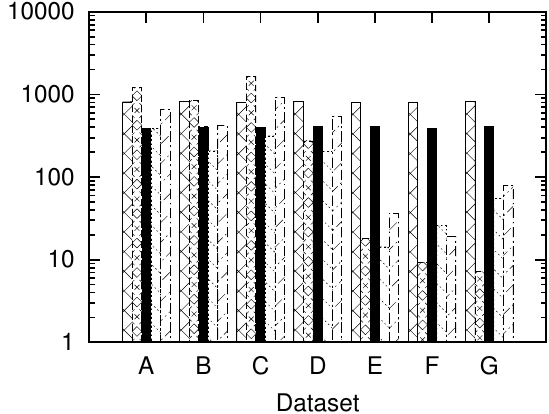}
    \caption{\label{fig:hist_0.1} \hist $(\epsilon = 0.1, G^1_k)$}
  \end{subfigure}
  \begin{subfigure}[t]{.23\textwidth}
    \centering
    \includegraphics[width=\textwidth]{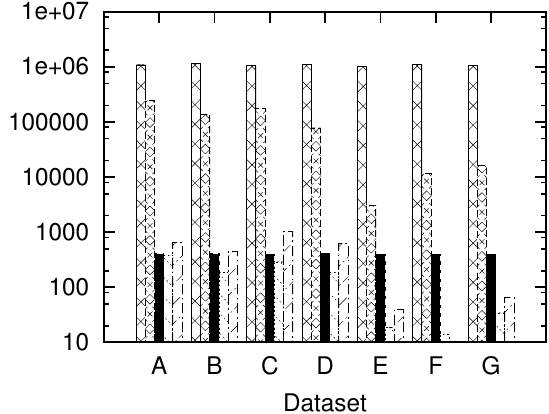}
    \caption{\label{fig:1Drange_0.1} \range  $(\epsilon = 0.1, G^1_k)$}  
  \end{subfigure}
  \begin{subfigure}[t]{.23\textwidth}
    \centering
    \includegraphics[width=\textwidth]{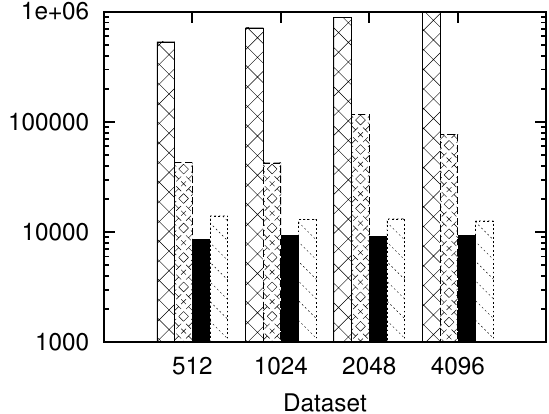}
    \caption{\label{fig:theta_0.1} \range  $(\epsilon = 0.1, G^4_k)$}  
  \end{subfigure}
  \caption{\label{fig:results} Comparison of $\epsilon/2$-Differentially private  and $(\epsilon,G)$-Blowfish algorithms for four workloads.}
\end{figure*}

In \csec\ref{sec:upperbound}, we outlined a number of algorithms for answering range query and marginal workloads under the distance threshold policies, and derived data independent error bounds for them. In this section, we implement both data independent and data dependent versions of these algorithms and empirically evaluate their error on a number of real one and two dimensional datasets. In particular, we compare the error attained by $\epsilon/2$-differentially private algorithm for a task to that of  $(\epsilon,G)$-Blowfish mechanisms for the same task. The highlights of this section are: 
\squishlist
\item For 1-D range queries, since the policy graphs are ``tree-like'', we can design data dependent algorithms for Blowfish by utilizing state-of-the-art data dependent algorithms for differential privacy, thanks to Theorem~\ref{thm:blowfishEquality}. The transformed workload is simpler and permits an order of magnitude improvement in error for both the data independent and data dependent implementations.
\item For 2-D, we are not aware of a low stretch embedding of the grid policy graph to a tree. Though we are restricted to the use of matrix mechanism algorithms, our new Blowfish private algorithms outperform the best data dependent differentially private algorithms on sparse datasets.
\item For histograms, while the transformed workload is not easier than the original, we can exploit constraints in the transformed database to obtain  improvements in error.
\squishend

\sstitle{Datasets:} We evaluate the error on 7 different one dimensional datasets and 3 two dimensional dataset (see Table~\ref{tab:datasets}). All the one dimensional datasets A-G have the same domain size (4096) but vary in their scale (total number of records), and were used in prior work (notably \cite{li14:dawa}). We aggregate our two dimensional dataset to a domain size of $100 \times 100$ (T100), $50\times 50$ (T50), and $25\times 25$ (T25), by imposing uniform grids of appropriate size on the space. Finally, we aggregate over dataset ``D'' to domain sizes 4096 (no aggregation), 2048, 1024, and 512. Note that most datasets are  sparse (low scale and high \% of zero counts). 

\sstitle{Policies:} We use $G^1_k$ and $G^4_k$ for one dimensional and $G^1_{k^2}$ for two dimensional datasets.

\sstitle{Workloads:} We consider four workloads. \range is 10,000 random one-dimensional range queries. \tworange is 10,000 random two dimensional range queries. \hist is the histogram workload. 
We report the average mean square error over 5 independent runs, and use $\epsilon  \in \{0.001, 0.01, 0.1, 1\}$.
Results for $\epsilon \in \{0.001,1\}$ are deferred to  \iftoggle{fullpaper}{Appendix~\ref{sec:more-experiments}}{the full version}.

\subsection{Results}
\sstitle{\hist}: We compare 5 algorithms on our 1-D datasets under policy graph $G^{1}_k$. We use Laplace mechanism and DAWA \cite{li14:dawa}, the best data independent and state-of-the-art data dependent differentially private algorithms for the workload, resp. For Blowfish, we use Laplace mechanism on the transformed database to get a data independent strategy (\csec\ref{sec:datadependent}). We called this ``Transformed + Laplace''. Since $G^1_k$ is a tree, we can construct two data dependent Blowfish algorithms as follows: (1) We use the consistency postprocessing algorithm to the output of Transformed + Laplace to ensure that the noisy counts $\tilde{\mathbf{x}}_G$ are non-decreasing. We call this ``Transformed + ConsistentEst''. (2) We compute a noisy histogram on $\mathbf{x}_G$ using DAWA and then apply consistency (called ``Transformed + Dawa + Cons''). 

We see that the $(\epsilon, G)$-Blowfish data independent technique (Transformed + Laplace) is only a factor of 2 better than the data independent $(\epsilon/2)$-differentially private algorithm (Laplace Mechanism). Significant gains in error over the data independent methods are seen in sparse datasets E, F \& G by using DAWA (for differential privacy) and the two data dependent Blowfish algorithms (note log scale on y-axis). When $\epsilon = 0.1$ (and $1$) one of the Blowfish data dependent algorithms outperform the two differentially private mechanism on all but two datasets F and G. These are very sparse, and DAWA achieves lower error. On the other hand, the transformed database (which is the set of prefix counts) is not as sparse, and thus Blowfish algorithms achieve higher error. At smaller $\epsilon$ values ($0.01, 0.001$), DAWA outperforms one or both the data dependent Blowfish algorithms on all datasets except $A$ and $B$. Designing data dependent Blowfish mechanism for \hist under $G^1_k$ with ``optimal'' error is an interesting open question. 

\sstitle{\range}: First, we consider the $G^1_k$ policy graph. We consider Privelet and DAWA as the data independent and dependent algorithms under differential privacy, resp. The Blowfish data independent strategy is to use Laplace mechanism on the transformed database. We again implement two data dependent algorithms for Blowfish -- enforce the non-decreasing constraint on the noisy $\tilde{\mathbf{x}}_G$ computed using the Laplace mechanism and DAWA, resp. In this case, we see 2-3 orders of magnitude difference in the error of all the Blowfish algorithms from their differentially private counterparts! This is because both the transformed workload and the transformed database are ``easier'' that the original workload and database. We observe that while the Blowfish data dependent strategy using DAWA is better than the one using Laplace mechanism on all datasets when $\epsilon = 1$,  the reverse is true for $\epsilon = 0.1, 0.01$. We conjecture this is true because the lower privacy budget results in poorer data dependent clustering for DAWA. 

We also study the error incurred under the $G^4_k$ policy graph under datasets of varying domain sizes ($k = 4096, 2048$, $1024, 512$). Since $G^4_k$ is not a tree, we use a spanning tree $H^4_k$ as described in Section~\ref{sec:range-1Dtheta} and Figure~\ref{fig:1Dtheta}. Since $G^4_k$ can be embedded into $H^4_k$ with a stretch of 3, by Corollary~\ref{cor:approx-equivalence} an $\epsilon/3$-differentially private mechanism for answering $\mathbf{W}_{H^4_k}$ on $\mathbf{x}_{H^4_k}$ also is a $(\epsilon, G^4_k)$-Blowfish private mechanism for answering $\mathbf{W}$ on $\mathbf{x}$. The Blowfish data independent algorithm is Laplace mechanism on the transformed database. The data dependent algorithm is DAWA on the transformed workload and database.  Both use privacy budget $\epsilon/3$. Again, the blowfish mechanisms have at least an order of magnitude smaller error than their differentially private counterparts.  While the error for the differentially private algorithms increases as the domain size increases, the error for the Blowfish mechanisms does not change with domain size. This is because the transformed workload is like the identity matrix. While the Blowfish data dependent algorithm is better than using Laplace mechanism for $\epsilon = 1$, it is worse for smaller $\epsilon$. 

\sstitle{\tworange}: We consider three algorithms for the grid graph policy $G^1_{k^2}$. Privelet and DAWA are the data independent and data dependent differentially private algorithms. The blowfish data independent strategy is to use Privelet for the one dimensional range queries in the transformed workload. We do not know of a data dependent algorithm under Blowfish for $G^1_{k^2}$, since it is not ``tree-like''. Though each transformed query requires 4 one dimensional range queries (over the transformed database), we still see that the Blowfish algorithm (a) significantly outperforms Privelet, and (b) improves over DAWA when the domain size is large. 

\eat{
<<<<<<< HEAD
=======
\sstitle{\marg}: We consider two data-independent strategies, one which is differentially private, and one under Blowfish privacy.
Under differential privacy, we divide the $\epsilon$-budget 2 ways, and then add independent Laplace noise to each query in each dimension.
All queries in a given dimension are independent.
Under Blowfish privacy, each query is independent, and we add independent Laplace noise to each.
By the results of Section~\ref{sec:1dmarginal}, we could also enforce ordering constraints on the answers to these queries.
Since our domain size was small, the noisy answers to the queries were already consistent, so enforcing consistency had no effect.
On larger domains, enforcing consistency would improve error dramatically, as seen in the strategies for histogram and 1D range queries.
We were not able to run experiments on larger domains due to computational constraints.
>>>>>>> 2beb17f73e31cebb5eea8cc7dabf43433ed733bf
}


\vspace{-2mm}
\section{Conclusions}
\label{sec:conclusions}
We systematically analyzed error bounds on linear  query workloads under the Blowfish privacy framework. We showed that the error incurred when answering a workload under Blowfish is identical to the error incurred when answering a transformed workload under differential privacy for a large class of privacy mechanisms and graphs. This, in conjunction with a subgraph approximation result, helped us derive strategies for answering linear counting queries under the Blowfish privacy framework. We showed that workloads can be answered with significantly smaller amounts of error per query under Blowfish privacy compared to differential privacy, suggesting the applicability of Blowfish privacy policies in practical utility driven applications. 

\stitle{Acknowledgements:}
We thank the anonymous reviewers for their comments. This work was supported by the National Science Foundation under Grants 1253327, 1408982, 1443014 and a gift from Google.

{\small 
\bibliographystyle{abbrv}
\bibliography{paper} 

\begin{thebibliography}{10}

\bibitem{andres2013geo}
M.~E. Andr{\'e}s, N.~E. Bordenabe, K.~Chatzikokolakis, and C.~Palamidessi.
\newblock Geo-indistinguishability: Differential privacy for location-based
  systems.
\newblock In {\em ACM CCS}, pages 901--914. ACM, 2013.

\bibitem{bhaskara2012unconditional}
A.~Bhaskara, D.~Dadush, R.~Krishnaswamy, and K.~Talwar.
\newblock Unconditional differentially private mechanisms for linear queries.
\newblock In {\em STOC}, 2012.

\bibitem{pets13:metric}
K.~Chatzikokolakis, M.~Andrés, N.~Bordenabe, and C.~Palamidessi.
\newblock Broadening the scope of differential privacy using metrics.
\newblock In {\em Privacy Enhancing Technologies}. 2013.

\bibitem{icalp:Dwork06}
C.~Dwork.
\newblock Differential privacy.
\newblock In {\em ICALP}, 2006.

\bibitem{dwork2012fairness}
C.~Dwork, M.~Hardt, T.~Pitassi, O.~Reingold, and R.~Zemel.
\newblock Fairness through awareness.
\newblock In {\em TCS}, pages 214--226. ACM, 2012.

\bibitem{tcc:DworkMNS06}
C.~Dwork, F.~McSherry, K.~Nissim, and A.~Smith.
\newblock Calibrating noise to sensitivity in private data analysis.
\newblock In {\em TCC}, pages 265--284, 2006.

\bibitem{stoc:DworkNPR10}
C.~Dwork, M.~Naor, T.~Pitassi, and G.~N. Rothblum.
\newblock Differential privacy under continual observation.
\newblock In {\em STOC}, pages 715--724, 2010.

\bibitem{fakcharoenphol2003tight}
J.~Fakcharoenphol, S.~Rao, and K.~Talwar.
\newblock A tight bound on approximating arbitrary metrics by tree metrics.
\newblock In {\em Proceedings of the thirty-fifth annual ACM symposium on
  Theory of computing}, pages 448--455. ACM, 2003.

\bibitem{stoc:HardtT10}
M.~Hardt and K.~Talwar.
\newblock On the geometry of differential privacy.
\newblock In {\em STOC}, pages 705--714, 2010.

\bibitem{vldb:HayRMS10}
M.~Hay, V.~Rastogi, G.~Miklau, and D.~Suciu.
\newblock Boosting the accuracy of differentially-private queries through
  consistency.
\newblock In {\em PVLDB}, pages 1021--1032, 2010.

\bibitem{blowfish}
X.~He, A.~Machanavajjhala, and B.~Ding.
\newblock Blowfish privacy: Tuning privacy-utility trade-offs using policies.
\newblock In {\em SIGMOD}, 2014.

\bibitem{sigmod:KiferM11}
D.~Kifer and A.~Machanavajjhala.
\newblock No free lunch in data privacy.
\newblock In {\em SIGMOD}, pages 193--204, 2011.

\bibitem{pods:KiferM12}
D.~Kifer and A.~Machanavajjhala.
\newblock A rigorous and customizable framework for privacy.
\newblock In {\em PODS}, 2012.

\bibitem{li14:dawa}
C.~Li, M.~Hay, and G.~Miklau.
\newblock A data- and workload-aware algorithm for range queries under
  differential privacy.
\newblock {\em To appear Proc. VLDB Endow.}, 2014.

\bibitem{pods:LiHRMM10}
C.~Li, M.~Hay, V.~Rastogi, G.~Miklau, and A.~McGregor.
\newblock Optimizing histogram queries under differential privacy.
\newblock In {\em PODS}, pages 123--134, 2010.

\bibitem{edbt:chaoli13}
C.~Li and G.~Miklau.
\newblock Optimal error of query sets under the differentially-private matrix
  mechanism.
\newblock In {\em ICDT}, 2013.

\bibitem{linial1995geometry}
N.~Linial, E.~London, and Y.~Rabinovich.
\newblock The geometry of graphs and some of its algorithmic applications.
\newblock {\em Combinatorica}, 15(2):215--245, 1995.

\bibitem{ashwin11:vldb}
A.~Machanavajjhala, A.~Korolova, and A.~D. Sarma.
\newblock Personalized social recommendations - accurate or private?
\newblock In {\em PVLDB}, volume~4, pages 440--450, 2011.

\bibitem{nikolov2013geometry}
A.~Nikolov, K.~Talwar, and L.~Zhang.
\newblock The geometry of differential privacy: the sparse and approximate
  cases.
\newblock In {\em ACM STOC}, pages 351--360. ACM, 2013.

\bibitem{icde:XiaoWG10}
X.~Xiao, G.~Wang, and J.~Gehrke.
\newblock Differential privacy via wavelet transforms.
\newblock In {\em ICDE}, pages 225--236, 2010.

\end{thebibliography}
}

\iftoggle{fullpaper}{
\clearpage 

\appendix

\begin{figure*}[th!]
  \captionsetup[subfigure]{width=\textwidth}
  \centering
  \begin{subfigure}[t]{.23\textwidth}
    \centering
    \includegraphics[width=\textwidth]{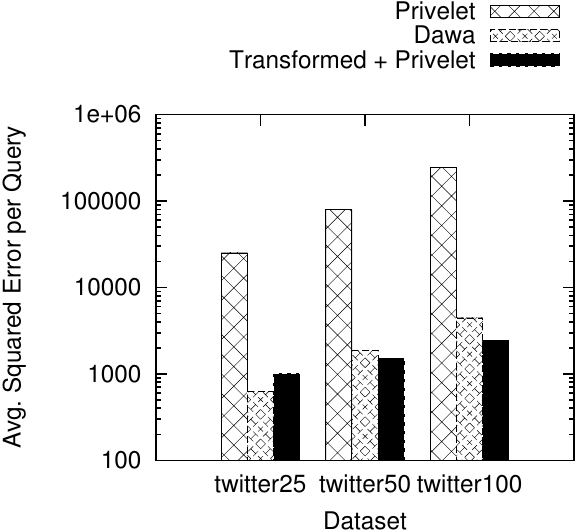}
    \caption{\label{fig:2Drange_1} \tworange $(\epsilon = 1, G^1_{k^2})$}  
  \end{subfigure}
  \begin{subfigure}[t]{.23\textwidth}
    \centering
    \includegraphics[width=\textwidth]{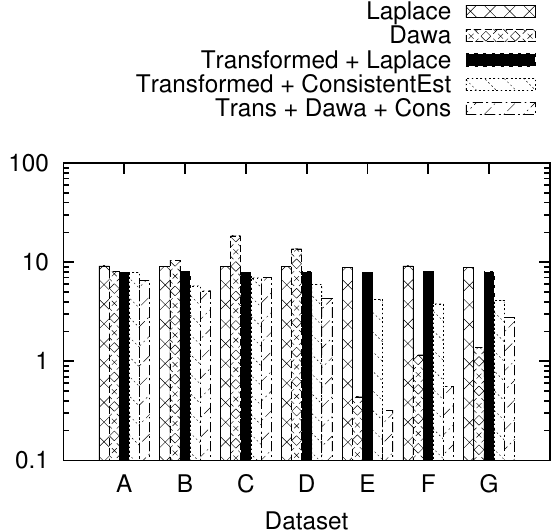}
    \caption{\label{fig:hist_1} \hist $(\epsilon = 1, G^1_k)$}
  \end{subfigure}
  \begin{subfigure}[t]{.23\textwidth}
    \centering
    \includegraphics[width=\textwidth]{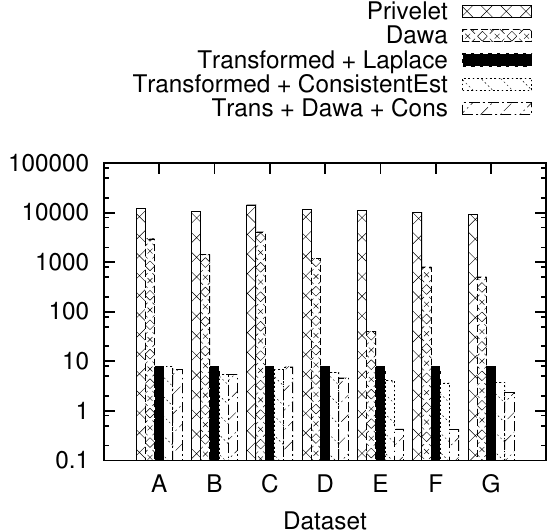}
    \caption{\label{fig:1Drange_1} \range  $(\epsilon = 1, G^1_k)$}  
  \end{subfigure}
  \begin{subfigure}[t]{.23\textwidth}
    \centering
    \includegraphics[width=\textwidth]{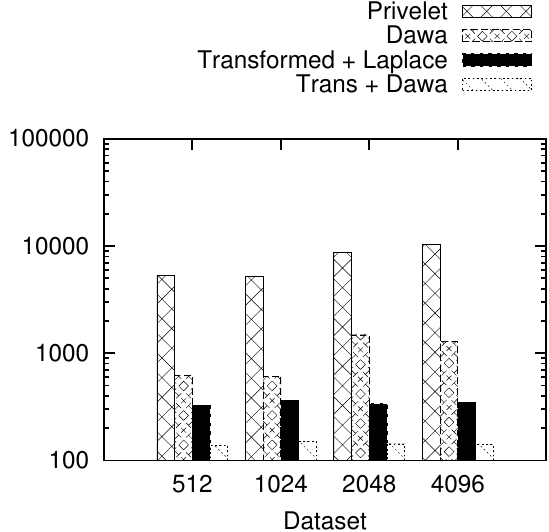}
    \caption{\label{fig:theta_1} \range  $(\epsilon = 1, G^4_k)$}  
  \end{subfigure}

  \begin{subfigure}[t]{.23\textwidth}
    \centering
    \includegraphics[width=\textwidth]{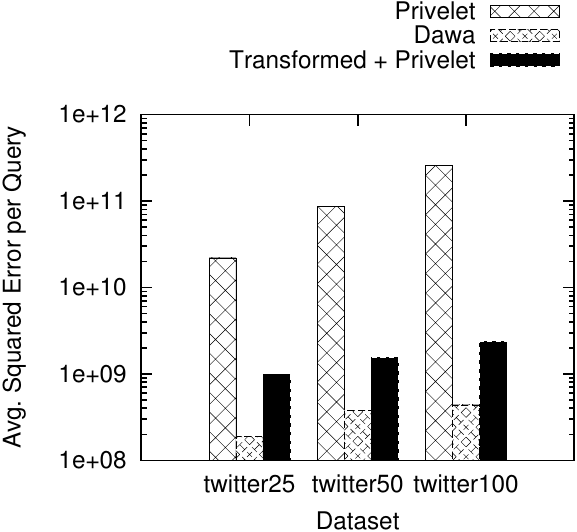}
    \caption{\label{fig:2Drange_0.001}  \tworange  $(\epsilon = 0.001, G^1_{k^2})$}  
  \end{subfigure}
  \begin{subfigure}[t]{.23\textwidth}
    \centering
    \includegraphics[width=\textwidth]{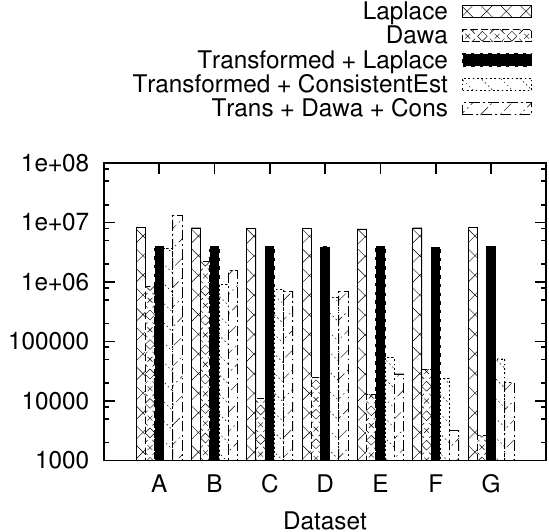}
    \caption{\label{fig:hist_0.001} \hist $(\epsilon = 0.001, G^1_k)$}
  \end{subfigure}
  \begin{subfigure}[t]{.23\textwidth}
    \centering
    \includegraphics[width=\textwidth]{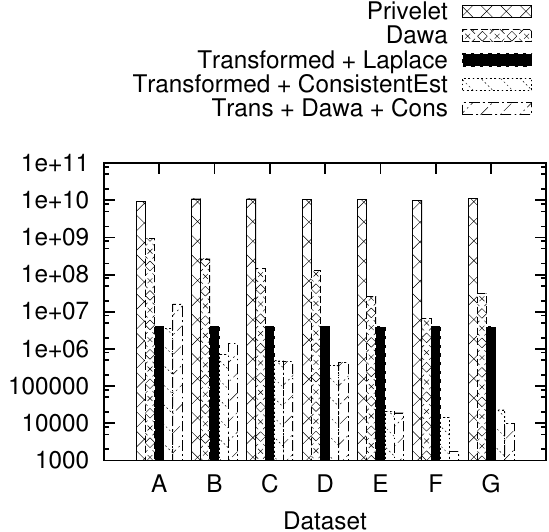}
    \caption{\label{fig:1Drange_0.001} \range  $(\epsilon = 0.001, G^1_k)$}  
  \end{subfigure}
  \begin{subfigure}[t]{.23\textwidth}
    \centering
    \includegraphics[width=\textwidth]{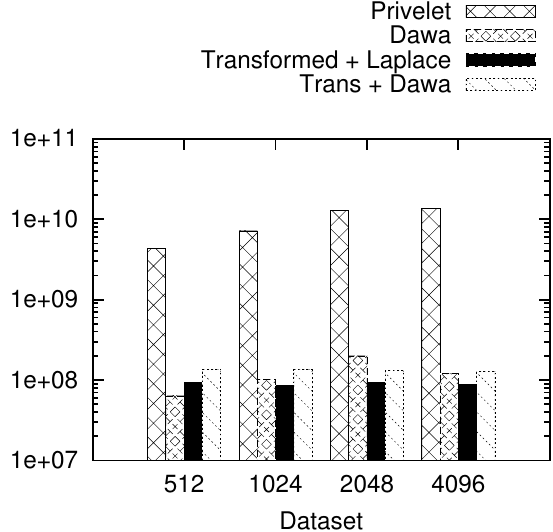}
    \caption{\label{fig:theta_0.001} \range  $(\epsilon = 0.001, G^4_k)$}  
  \end{subfigure}
  \caption{\label{fig:results-appendix} Comparison of $\epsilon/2$-Differentially private  and $(\epsilon,G)$-Blowfish algorithms for four workloads.}
\end{figure*}

\section{Extending Lower Bound Results in Differential Privacy}
\label{sec:lowerbound}

In this section, we show that error upper and lower bound results in $\epsilon$-differential privacy and $(\epsilon,\delta)$-differential privacy can be extended to Blowfish privacy.
Any $\epsilon$-differentially private mechanism $\mathcal{M}_G$ for answering $\mathbf{W}_G$ on $\mathbf{x}_G$ gives us a $(\epsilon,G)$-Blowfish private mechanism on for answering $\mathbf{W}$ on $\mathbf{x}$.
The following corollary follows directly from Theorem~\ref{thm:blowfishEquality}.

\begin{corollary}
  Let $\mathbf{W}$ be a workload, $G$ be a policy graph, and $\mathbf{x}$ be a database.
  Let $f(\mathbf{W}_G,\mathbf{x}_G)$ be a lower bound on the error of answering $\mathbf{W}_G$ with $\mathbf{x}_G$ under $\epsilon$-differential privacy.
  That is,
  \begin{equation*}
    \mathrm{ERROR}_{\mathcal{M}}(\mathbf{W}_G,\mathbf{x}_G) = \Omega (f(\mathbf{W}_G,\mathbf{x}_G)) \quad \forall  \mathcal{M}.
  \end{equation*}
  Then, $f(\mathbf{W}_G,\mathbf{x}_G)$ is a lower bound for answering $\mathbf{W}$ with $\mathbf{x}$ under $(\epsilon,G)$-Blowfish privacy.
  \label{cor:blowfish-lower-bound}
\end{corollary}

Our transformational equivalence result applies to every $\epsilon$-differentially private mechanism  for answering linear queries when policy graphs look like trees. While we can't discuss every result in differential privacy here, we give a few examples.

One nice extension is to a sequence of works on general mechanisms for answering linear queries \cite{stoc:HardtT10, bhaskara2012unconditional, nikolov2013geometry}.
These results first find a lower bound (in this case, the lower bound is data independent, and therefore depends only on $\mathbf{W}$), then find a general mechanism that answers any linear workload with error bounded against the lower bound.
In particular, Bhaskara et al \cite{bhaskara2012unconditional} give a $O(\log^2 q)$ approximation to the expected $L_2$ error for linear workloads under $\epsilon$-differential privacy, where $q$ is the number of queries. For policy graphs $G$ that has a subtree $T$ with low distortion of $\ell$, we can use the same algorithm on $\mathbf{W}_T$ to get a $O(\ell^2 \cdot \log^2 q)$ approximation to the expected $L_2$ error for answering $\mathbf{W}$ under $(\epsilon,G)$-Blowfish privacy! This follows from Corollary~\ref{cor:approx-equivalence} and the fact that $\mathbf{W}_G$ is always linear as long as $\mathbf{W}$ is linear and the two workloads have the same number of queries.

A number of results pertain to $(\epsilon, \delta)$-differential privacy, which allows a small probability $\delta$ of failing the indistinguishability condition. We can similarly define $(\epsilon, \delta, G)$-Blowfish privacy, and our transformational equivalence result directly extends to this variant as well. Thus, we can also extend upper and lower bound results on $(\epsilon, \delta)$-differential privacy to Blowfish. 

Li and Miklau \cite{edbt:chaoli13} give a lowerbound for a popular class of mechanisms called \emph{matrix mechanisms} for answering workloads of linear queries.
Since transformational equivalence holds for all policy graphs for matric mechanisms, we can extend their bound for all Blowfish policies:

\begin{corollary}\label{cor:blowfishSVDB}
  Any matrix mechanism based strategy for answering workload $\mathbf{W}$ that satisfies $(\epsilon, \delta, G)$-Blowfish privacy has error at least
  \begin{equation*}
    P(\epsilon,\delta)\frac{1}{n_G}(\lambda_1+\ldots+\lambda_s)^2 
  \end{equation*}
  where $P(\epsilon,\delta) = \frac{2\log(2/\delta)}{\epsilon^2}$, $\lambda_1,\ldots,\lambda_s$ are the singular values of $\mathbf{W}_G$, and $n_G$ is the number of columns of $\mathbf{W}_G$ (same as the number of edges in $G$).
\end{corollary}
\eat{
\begin{proof}
  From the results of \cite{edbt:chaoli13} we have 
  \begin{equation*}
    \mathrm{ERROR}_\mathcal{M}(\mathbf{W}_G)\ge P(\epsilon,\delta)\frac{1}{n_G}(\lambda_1+\ldots+\lambda_s)^2
  \end{equation*}
  for any matrix mechanism $\mathcal{M}$, where $P(\epsilon,\delta) = \frac{2\log(2/\delta)}{\epsilon^2}$, $\lambda_1,\ldots,\lambda_s$ are the singular values of $\mathbf{W}_G$, and $n_G$ is the number of columns of $\mathbf{W}_G$. 
  The result then follows immediately from Corollary~\ref{cor:blowfish-lower-bound}.
\end{proof}
} 


Figures~\ref{fig:allrange-lb} and \ref{fig:2dranges-lb} illustrate the relationship between the above lower bound on error and size of the domain for $\mathbf{R}_k$ (under $G^\theta_k$) and $\mathbf{R}_{k^2}$ (under $G^\theta_{k^2}$) respectively. We plot the original lower bound for unbounded differential privacy  (\cite{edbt:chaoli13}) and the new lower bounds we derived for Blowfish policies $G_k^{\theta}$ and $G^\theta_{k^2}$ for various values of $\theta$. Additionally, we show a lower bound for bounded differential privacy, which is obtained by using the complete graph (on $\dom$) as the policy. 

For the one dimensional range query workload we see that minimum error under unbounded differential privacy increases faster than the minimum error under $G^{\theta}_k$ for sufficiently large domain sizes.  For two dimensional ranges, error under Blowfish policy $G^{\theta}_{k^2}$ is only better than unbounded differential privacy for $\theta=1$. However, all values of $\theta$ perform better than bounded differential privacy. Note that for sets of linear queries, it is possible for the sensitivity of a workload under bounded differential privacy to be twice the sensitivity of the workload under unbounded differential privacy, and thus have up to 4 times more error. Characterizing analytical lower bounds for these workloads and policies is an interesting avenue for future work.

\begin{figure}[t]
  \centering
  \begin{subfigure}{.23\textwidth}
    \centering
    \includegraphics[width=\textwidth]{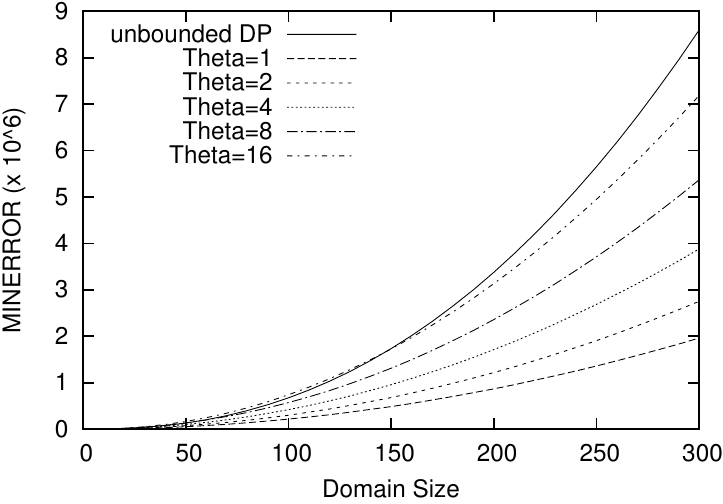}
    \caption{1D Ranges under $G^{\theta}_k$} 
    \label{fig:allrange-lb}
  \end{subfigure}
  \begin{subfigure}{.23\textwidth}
    \centering
    \includegraphics[width=\textwidth]{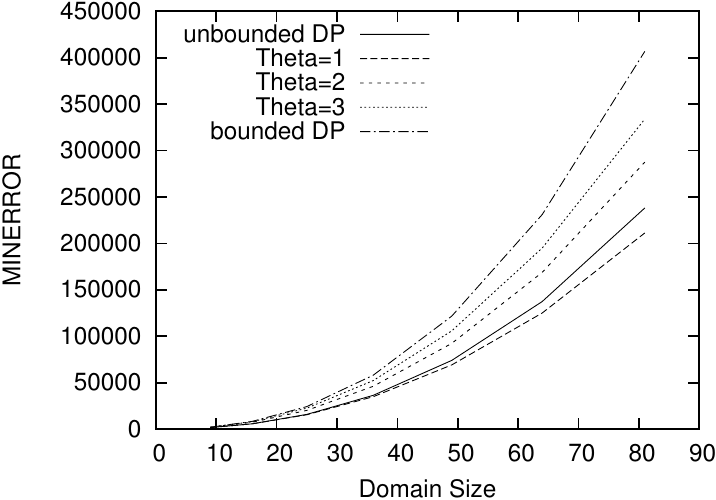}
    \caption{$2$D Ranges under $G^{\theta}_{k^2}$}
    \label{fig:2dranges-lb}
  \end{subfigure}
  \caption{Blowfish SVD lower bounds ($\epsilon = 1, \delta = .001$).}
  \label{fig:lb}
\end{figure}


\section{More Experiments}
\label{sec:more-experiments}
Figure~\ref{fig:results-appendix} compares our Blowfish algorithms to the corresponding differentially private counterparts for $\epsilon \in \{0.001, 1.0\}$.

\section{Proofs from Section 4.3}\label{sec:proofs-general-transequiv}
\subsection{Impossiblity of General Transformational Equivalence}
\label{sec:imposs-general-transequiv}
\begin{theorem}
Let $G$ be a graph that does not have an isometric embedding into points in $L_1$. There exists a mechanism $\mathcal{M}$ and workload $\mathbf{W}$ such that for any transformation of $(\mathbf{W}, \mathbf{x}) \rightarrow (\mathbf{W}_G, \mathbf{x}_G)$ such that $\mathbf{W}\mathbf{x} = \mathbf{W}_G\mathbf{x}_G$, either $\mathcal{M}$ is not an $(\epsilon, G)$-Blowfish private mechanism for answering $\mathbf{W}$ on $\mathbf{x}$, or $\mathcal{M}$ is not a $\epsilon$-differentially private mechanism for answering $\mathbf{W}_G$ on $\mathbf{x}_G$.
\end{theorem}
\begin{proof}
  Let $\mathbf{W}$ be the identity workload $\mathbf{I}$. 
  We assume that policy graph $G$ cannot be isometrically embedding into $L_1$ (e.g., the cycle on $V$).
  We consider datasets $\mathbf{x}$ with a single entry. Thus there are $|V|$ input databases, and each input corresponds to a vertex in $G$; hence, we will abuse notation and denote by $\mathbf{x}$ both the dataset and the vertex in $G$. There are also $|V|$ distinct outputs since we consider the identity workload. 
  Let $\mathbf{x}_G$ be the point in the $L_1$ that the vertex $\mathbf{x}$ is mapped to. 
  We denote distance in $G$ as $d_G(\cdot,\cdot)$, and $L_1$ distance as $d(\cdot,\cdot)$.
  We consider two cases:
  \begin{enumerate} 
    \item\label{case1} there exists $\mathbf{z},\mathbf{y}$ (each with a single entry) with  $d(\mathbf{z}_G,\mathbf{y}_G) = 1$ and $d_G(\mathbf{z},\mathbf{y}) > 1$.
   \item\label{case2} There exists $\mathbf{z},\mathbf{y}$ (each with a single entry) with  $d_G(\mathbf{z},\mathbf{y}) = 1$ and $d(\mathbf{z}_G,\mathbf{y}_G) > 1$, or
  \end{enumerate}

\stitle{Case (\ref{case1}):} We will now construct a mechanism $\mathcal{M}$ that is $(\epsilon, G)$-Blowfish private, but not $\epsilon$-differentially private.
  Let $\mathcal{M}$ be the exponential mechanism that given an input $\mathbf{x}$ picks an output $\mathbf{y}$ with probability $e^{-\epsilon \cdot d_G(\mathbf{x}, \mathbf{y})}$. It is easy to check that $\mathcal{M}$ satisfies $(\epsilon, G)$-Blowfish privacy. 

However, there exists $\mathbf{x},\mathbf{y}$ such that $\mathbf{x}_G$ and $\mathbf{y}_G$ are neighbors under differential privacy ($d(\mathbf{x}_G,\mathbf{y}_G) = 1$), but $\mathbf{x}$ and $\mathbf{y}$ are not neighbors under policy graph $G$ (since $d_G(\mathbf{z},\mathbf{y}) > 1$). Therefore, $\mathcal{M}$ does not satisfy differential privacy: 
\begin{equation}
\frac{P[\mathcal{M}(\mathbf{x}_G = \mathbf{x}]}{P[\mathcal{M}(\mathbf{x}_G = \mathbf{y}]}  =  e^{\epsilon \cdot d_G(\mathbf{z},\mathbf{y})} > e^\epsilon
\end{equation}

\stitle{Case (\ref{case2}):} Proof is similar. We construct $\mathcal{M}$ as the exponential mechanism that uses the distances on $\mathbf{x}_G$  as the score function. 
\end{proof}

\subsection{Subgraph Approximation}
\begin{lemma}(Subgraph Approximation)
  Let $G = (V,E)$ be a policy graph. Let $G' = (V, E')$ be a spanning tree of $G$ on the same set of vertices, such that every $(u,v) \in E$ is connected in $G'$ by a path of length at most $\ell$ ($G'$ is said to be an $\ell$-approximate subgraph\footnote{While we that require $V(G) = V(G')$, the proof does not require $G'$ to be a subgraph of $G$ (i.e., $E' \subseteq E$). But it suffices for the applications of this technique in this paper.}). Then for any mechanism $\mathcal{M}$ which satisfies $(\epsilon,G')$-Blowfish privacy, $\mathcal{M}$ also satisfies $(\ell \cdot \epsilon,G)$-Blowfish privacy.
\end{lemma}
\begin{proof}
  Assume $D$ and $D'$ are neighboring databases under policy graph $G$. Then $D = A \cup \left\{ x \right\}$ and $D' = A \cup \left\{ y \right\}$ for some database $A$, and $(x,y) \in E$. From our assumption, $x$ and $y$ are connected by a path in $G'$ of length at most $\ell$. Therefore, there exist a sequence of vertices $x=v_1,\dots,v_j = y$ such that $(v_i, v_{i+1}) \in E$ and $j<\ell$. Further, $A \cup \left\{ v_i \right\}$ and $A \cup \left\{ v_{i+1} \right\}$ are neighbors under policy graph $G'$. Therefore, 
  \begin{equation*}
    \mathrm{Pr} [\mathcal{M}(A \cup \left\{ v_i \right\}) \in S] \le e^\epsilon \cdot \mathrm{Pr} [\mathcal{M}(A \cup \left\{ v_{i+1} \right\}) \in S].
  \end{equation*}
    Composing over all $1 \le i \le j$ gives us the desired result. 
\end{proof}

\section{Properties of $\mathbf{P}_G$}
\label{sec:PG-appendix}
\begin{lemma}
  Let $\mathbf{W}$ be a workload, and $G$ be a policy graph. Then $\Delta_\mathbf{W}(G) = \Delta_{\mathbf{W}_G}$.
\end{lemma}
{\sc Proof.}
  This follows from the definition of $\mathbf{P}_G$. We have 
  \begin{align*}
    \Delta_{\mathbf{W}}(G) &= \max_{(\mathbf{x},\mathbf{x}')\in N(G)}\norm{\mathbf{Wx} - \mathbf{Wx}'}_1  \\
    &= \max_{\mathbf{v}_i\in\text{cols}(\mathbf{W}_G)}\norm{\mathbf{v}_i}_1 \rlap{$\qquad \Box$}
  \end{align*}

\begin{lemma}
$\mathbf{P}_G$ constructed in Section~\ref{sec:construction-PG} has rank $k$.
\end{lemma}
\begin{proof}
It is sufficient for us to show that if $G$ is a tree, then $\mathbf{P}_G$ has rank $k$. When a connected graph $G$ is not a tree, we can consider a spanning tree $T$ of $G$. If the columns of $\mathbf{P}_G$ that correspond to the edges in $T$ have rank $k$, then $\mathbf{P}_G$ will also have rank $k$. 

A spanning tree has $k+1$ vertices (including $\bot$) and $k$ edges. Thus, $\mathbf{P}_T$ is a $k \times k$-matrix. It is sufficient to prove that the linear system $\mathbf{P}_T \mathbf{y} = \mathbf{0}$ has an unique solution $\mathbf{y} = \mathbf{0}$. Each entry in $\mathbf{y}$ can be considered as a weight $\mathbf{y}[e]$ associated with an edge $e$ in the tree $T$. 

Suppose there is an edge $e$ connected to a degree 1 node $v \in T$ that is not connected to $\bot$. The row $\mathbf{P}_T[v, :]$ has only one non-negative entry in the column corresponding to $e$. Thus, $\mathbf{P}_T\mathbf{y}  =0$ implies $\mathbf{y}[e]=0$. We can inductively use the same argument on the tree $T'(V - \{v\}, E - \{e\})$, and the matrix $\mathbf{P}_{T'}$ constructed by removing the row corresponding to node $v$ from $\mathbf{P}_T$. In the end, we are left with a tree $T^\star$ having only edges connected to $\bot$. But this would correspond to a matrix $\mathbf{P}_{T^\star} = \mathbf{I}$ (since each column has exactly one position set to 1), and the solution to $\mathbf{I}\mathbf{y} = \mathbf{0}$ is $\mathbf{y} = \mathbf{0}$.
\end{proof}

\begin{lemma}
  Suppose $\mathbf{P}_G$ is constructed for a Blowfish policy graph $G$ as above, and $G$ is a tree. Any pair of databases $\mathbf{y}, \mathbf{z} \in \mathbb{R}^k$ are neighbors according to the Blowfish policy $G$ if and only if $\mathbf{P}_G^{-1}\mathbf{y}$ and $\mathbf{P}_G^{-1}\mathbf{z}$ are neighboring databases according to unbounded differential privacy.
\end{lemma}
\begin{proof}
Let $\mathbf{y}' = \mathbf{P}_G^{-1}\mathbf{y}$ and $\mathbf{z}' = \mathbf{P}_G^{-1}\mathbf{z}$. Consider $\mathbf{y} - \mathbf{z} = \mathbf{P}_G\mathbf{y}'-\mathbf{P}_G\mathbf{z}'$. Since $\mathbf{y}'$ and $\mathbf{z}'$ are neighbors under differential privacy, $\mathbf{y}'-\mathbf{z}' = \mathbf{\hat{i}}$, where $\mathbf{\hat{i}}$ is a vector with a single non-zero entry, which is 1. So, $\mathbf{y} - \mathbf{z} = \mathbf{P}_G\mathbf{\hat{i}}$ is a single column of $\mathbf{P}_G$. We have, either
\squishlist 
\item $\mathbf{y} - \mathbf{z}$ is equal to a column of $\mathbf{P}_G$ corresponding to $(u,v) \in E$ where $u,v \neq \bot$, or
\item $\mathbf{y} - \mathbf{z}$ is equal to a column of $\mathbf{P}_G$ corresponding to $(u,\bot) \in E$ where $u \neq \bot$.
\squishend
In either case, $\mathbf{y}$ and $\mathbf{z}$ are neighbors according to $G$. 

For the proof in the other direction, suppose $\mathbf{y}$ and $\mathbf{z}$ are neighbors, that is, $\mathbf{y} - \mathbf{z} = \mathbf{p}_G$ where $\mathbf{p}_G$ is a column of $\mathbf{P}_G$.
In this direction, we will use that $G$ is a tree, and therefore $\mathbf{P}_G$ is square, which implies $\mathbf{P}_G^{-1}$ is both a left and right inverse.
We can write
\begin{eqnarray*}
  \lefteqn{\mathbf{P}_G\mathbf{P}_G^{-1}\mathbf{y} - \mathbf{P}_G\mathbf{P}_G^{-1}\mathbf{z} = \mathbf{p}_G \implies} \\
  && \mathbf{P}_G^{-1} \cdot \mathbf{P}_G(\mathbf{P}_G^{-1}\mathbf{y} - \mathbf{P}_G^{-1}\mathbf{z}) = \mathbf{P}_G^{-1} \cdot \mathbf{p}_G 
  = \mathbf{\hat{i}},
\end{eqnarray*}
for some unit vector $\mathbf{\hat{i}}$.
\end{proof}

\subsection{Technical Details for Case II}\label{sec:caseII-PG-appendix}
More formally, assuming $v$ is the $i$th value, let
\[
\mathbf{W}' = \mathbf{W} \mathbf{D} = \mathbf{W} \left(
\begin{array}{ccc}
\mathbf{I}_{i-1} &  \mathbf{0}
\\
-\mathbf{1}^\top_{i-1} &  -\mathbf{1}^\top_{k-i}
\\
\mathbf{0} &  \mathbf{I}_{k-i}
\end{array}
\right),
\]
where $\mathbf{I}_j$ is an identity matrix with size $j$, $\mathbf{1}_j$ is a $j$-dim vector with all entries equal to $1$, and thus $\mathbf{D}$ is $k\times(k-1)$-matrix.
The following lemma shows requirements (b) and (c) to complete our construction.

\begin{lemma}
Consider $G'$, $\mathbf{W}'$, and $\mathbf{x}_{-v}$ constructed above. We have: i) $\mathbf{W} \mathbf{x} = \mathbf{W}' \mathbf{x}_{-v} + \mathbf{c}(\mathbf{W}, n)$, where $\mathbf{c}(\mathbf{W}, n)$ is a constant vector depending only on $\mathbf{W}$ and the size of the database; and ii) any two databases $\mathbf{y}$ and $\mathbf{z}$ are neighbors under $G$ if and only if $\mathbf{y}_{-v}$ and $\mathbf{z}_{-v}$ are neighbors under $G'$.
\end{lemma}
\begin{proof}
  We define $\mathbf{c}(\mathbf{W}, n)$ to be a vector with values $-n$ for the entries corresponding to affected rows of $\mathbf{W}$, and zero elsewhere.
  Consider a query $\mathbf{q}'$ in $\mathbf{W}'$ which depends on $\mathbf{v}$, the column of $\mathbf{W}$ no longer present in $\mathbf{W}'$.
  That is, $\mathbf{q}'$ corresponds to an affected row of $\mathbf{W}$, and we denote by $\mathbf{q}$ the original query in $\mathbf{W}$.
  By definition of $\mathbf{W}'$, $\mathbf{q}'\mathbf{x}'$  is the answer to $n-\mathbf{qx}$.
  
  Proof of ii):
  Suppose $\mathbf{z}$ and $\mathbf{y}$ are neighbors under $G$.
  We consider two cases:
  \begin{enumerate}
    \item $\mathbf{z}$ and $\mathbf{y}$ differ in two entries adjacent in $G$, neither of which is $v$.
    \item $\mathbf{z}$ and $\mathbf{y}$ differ in two entries adjacent in $G$, one of which is $v$.
  \end{enumerate}
  In the former case, $\mathbf{y}_{-v}$ and $\mathbf{z}_{-v}$ differ in the same two entries, and these vertices are still adjacent in $G'$.
  Therefore they are still neighbors.
  In the latter case, $\mathbf{z}_{-v}$ and $\mathbf{y}_{-v}$ differ in a single entry, $w$.
  $w$ was originally connected to $v$ in $G$ and is therefore connected to $\bot$ in $G'$.
  Therefore, $\mathbf{z}_{-v}$ and $\mathbf{y}_{-v}$ are neighbors according to $G'$.
\end{proof}

Using the above lemma, for any given workload $\mathbf{W}$ on database $\mathbf{x}$ and policy graph $G$ in this case, we can first convert them into $\mathbf{W}'$, $\mathbf{x}'$, and $G'$ as discussed above, apply the transformation equivalence theorem using $\mathbf{P}_{G'}$. We can answer the original workload using answers from $\mathbf{W}'\mathbf{x}'$  using i) in the above lemma.

\section{Blowfish Policies with Multiple Connected Components}
\label{sec:multiple-components}
Policy graphs can be disconnected as described in the example below. 

\noindent({\em Sensitive Attributes in Relational Tables}): Consider tables with $d$ attributes, \ie, $\dom =A_1 \times \ldots \times A_d$. Suppose ${\cal S} \subsetneq \{A_1, \ldots, A_d\}$ is a set of sensitive attributes (e.g., disease status and race) we want to protect in a table. In the corresponding policy graph $G(V,E)$, we have $V = \dom$, and $(u,v) \in E$ for any pair of $u,v \in \dom$ if and only if $u$ and $v$ differ in exactly one of the attributes in ${\cal S}$.

In the  above policy,  there is no path between $u$ and $v$ if they differ in the value of an attribute not in ${\cal S}$. In this case, there is no bound on probabilities when considering database $D \cup \{u\}$ and $D \cup \{v\}$. In particular, if $G$ has $c$ connected components $C_1, \ldots, C_c$, $C_i = (V_i, E_i)$, the adversary can determine (exactly) which $V_i$ every tuple in the database belongs to. In the ``sensitive attributes'' policy, the adversary can learn the values of the non-sensitive attributes of all the tuples. Blowfish provides the flexibility of such exact disclosure based on the level of privacy we want to protect. Exact disclosure never occurs in differential privacy or with connected Blowfish policies.
One potential issue with such exact disclosure is that further privacy disclosure may occur if adversaries know additional correlations between the disclosed properties (e.g., non-sensitive attributes) and private properties (e.g., sensitive properties). This is akin to privacy disclosures from differentially private outputs to adversaries with knowledge of correlations across records in the data \cite{sigmod:KiferM11,pods:KiferM12}. 

\subsection{Transformational Equivalence for Disconnected Policies}
The transformational equivalence result also extends to general policy graphs with more than one connected component. If $\bot$ is not connected to a component, we simply apply the conversion discussed in Case II to reduce it to Case I with the vertex $\bot$ connected to it. Then eventually we will have all components connected to $\bot$, which essentially falls into Case I. So Case III can be also reduced to Case I.

\section{Proof of Theorem 5.6}

\begin{proof}
  We first decompose our query into two pieces: all internal edges, and all external edges. We find strategies to answer each of these queries, then sum the two to find the answer to the desired query. Figure~\ref{fig:2Dtheta-c} shows the set of external edges in the transformed query. External edges always form a lattice, so we can answer this part of the query using the strategy from Section~\ref{sec:range-2Dline}, and this will contribute $O(d \frac{\log^{3(d-1)} {d \cdot k}/\theta}{\epsilon^2})$ error.
  
  We also need a way to answer all the internal edges. We order these edges by their black endpoint. Consider the set of vertices, $V$ corresponding to the set of internal edges which satisfy Lemma~\ref{lem:query-transformation}. $V$ can be divided into $2d$ $d$-dimensional range queries, one for each face of the original range query. These $d$-dimensional range queries are bounded by $\theta$ in the dimension orthogonal to the corresponding face of the original range query. This is illustrated in two dimensions in Figure~\ref{fig:2Dtheta-d}. Our strategy to answer these bounded ranges is the following: For each dimension, divide the domain (which is a hypercube of size $k^d$) into $\frac{d \cdot k}{\theta}$ layers, each with thickness $\theta/d$. We then answer all range queries on each layer. For a given dimension, all layers are independent. Therefore, we can answer these sets of range queries in parallel. However, the sets of range queries for different dimensions are not independent. An edge used in some horizontal layer will also be used in some vertical layer. We can answer each set of range queries using the Privelet framework with error
  \begin{equation*}
    O(\frac{log^{3(d-1)} k \log^3 \theta/d}{\epsilon^2}).
  \end{equation*}
  However, because range queries in different dimensions are dependent, we must divide up our $\epsilon$-budget $d$ ways. Additionally, each query is made up of $2d$ of these range queries. The total error of this strategy is therefore
  \begin{equation*}
    \label{equ:2Dtheta-error}
    O(d^3 \cdot \frac{log^{3(d-1)} k \log^3 \theta/d}{\epsilon^2}).
  \end{equation*}
  The total error is the sum of the errors from the strategies of answering the internal edges and the external edges. This sum is just Equation~\ref{equ:2Dtheta-error}.
\end{proof}

}

\end{document}